\documentclass[11pt,reqno]{amsart}

\date{\today}

\usepackage[ansinew]{inputenc}
\usepackage{textcomp}
\usepackage{cite}

\usepackage{amsthm}
\usepackage{amssymb}
\usepackage{amsfonts}

\newtheorem{theorem}{Theorem}[section]
\newtheorem{proposition}[theorem]{Proposition}
\newtheorem{lemma}[theorem]{Lemma}
\newtheorem{corollary}[theorem]{Corollary}

\theoremstyle{definition}

\theoremstyle{definition}
\newtheorem*{rem}{Remark}

\numberwithin{equation}{section}

\newcommand{\ac}{\textnormal{ac}}

\newcommand{\C}{\mathbb{C}}

\newcommand{\da}{\downarrow}
\newcommand{\E}{\mathbb{E}}

\newcommand{\lk}{\left(}

\newcommand{\N}{\mathbb{N}}
\newcommand{\p}{\mathbb{P}}

\newcommand{\R}{\mathbb{R}}
\newcommand{\rk}{\right)}
\newcommand{\s}{\textnormal{sc}}

\newcommand{\tr}{\textnormal{Tr}}

\DeclareMathOperator*{\esssup}{ess\,sup}
\DeclareMathOperator{\supp}{supp}

\begin{document}

\title[Density of states and eigenvectors on  random regular graphs]{Convergence of the density of states and delocalization of eigenvectors  \\ on random regular graphs}

\author[Leander Geisinger]{Leander Geisinger}
 \address{Leander Geisinger,  Princeton University, Princeton, NJ 08544, USA}
\email{leander.geisinger@gmail.com}


\begin{abstract}
Consider a random regular graph of fixed degree $d$ with $n$ vertices. We study spectral properties of the adjacency matrix and of random Schr\"odinger operators on such a graph as $n$ tends to infinity. 

We prove that the integrated density of states on the graph converges to the integrated density of states on the infinite regular tree and we give uniform bounds on the rate of convergence. This allows to estimate the number of eigenvalues in intervals of size comparable to $\log_{d-1}^{-1}(n)$.
Based on related estimates for the Green function we derive results about delocalization of eigenvectors.
\end{abstract}

\maketitle


\section{	Introduction}

In his seminal work from 1955, Wigner showed that the density of states of random matrices with independent identically distributed entries converges -- as the size of the matrix tends to infinity -- to a universal deterministic probability measure, the semicircle law \cite{Wig55}.
This result was gradually improved and today it is known that universality of spectral properties of Wigner matrices goes far beyond that.  Even the local eigenvalue statistics, studied via the eigenvalue gap distribution, is given by universal laws \cite{ErdPecRamSchYau10,TaoVu11,ErdRamSchTaoVauYau10,ErdYauYin12}. These laws can be calculated explicitly from Gaussian ensembles and are characterized by local repulsion of the eigenvalues. We refer to the textbooks \cite{Meh04,BaiSil10,AndGuiZei10} and the  review \cite{Dei07} for general results about random matrices and further references.

One field where random matrices arise is the study of random graphs. Along with a graph of $n$ labeled vertices one considers the adjacency matrix which is the symmetric $n \times n$ matrix with entry $ij$ equal to $1$, if vertex $i$ is connected by an edge to vertex $j$, and equal to $0$ otherwise. 
The spectrum of this matrix bears information  about the geometry of the graph, however determining spectral properties is often a difficult task. 

For Erd\H{o}s-R\'enyi graphs, where every edge is chosen independently with probability $p$ (see \cite{Bol01} for details concerning  random graph models),  universality of the local eigenvalue statistics for the corresponding adjacency matrix has recently been proved under suitable conditions on $p$ \cite{ErdKnoYauYin13,ErdKnoYauYin12}.
Two of the steps towards universality are proving local convergence of the density of states and delocalization of eigenvectors.  The former refers to the fact that the average number of eigenvalues in small intervals converges to a universal law.  (Ideally this holds for intervals of size comparable to $1/n$, up to logarithmic corrections.) The latter means that eigenvectors are typically uniformly distributed over the whole graph.
Let us emphasize that the mentioned results, for Wigner matrices as well as for Erd\H{o}s-R\'enyi graphs, rely on independence of the entries.

In this note we investigate the spectrum on random regular graphs with fixed degree $d$. In a regular graph of degree $d$ each vertex is connected to $d$ other vertices. Thus in the corresponding adjacency matrix each row and each column contains $d$ entries that are $1$ and $n-d$ entries that are $0$. Such a graph, or equivalently such a matrix, is chosen at random with uniform probability. These random matrices are sparse: only few entries are non-zero so that moments of the distribution of the entries decay slowly. More importantly the entries lack independence. Therefore it is not clear how to apply the methods developed to study Wigner matrices and Erd\H{o}s-R\'enyi graphs. 

While there are results concerning extreme eigenvalues at the edge of the spectrum \cite{Fri08,Sod09} not much seems to been known about eigenvalues in the bulk of the spectrum.
However, it is conjectured that the local statistics is universal and governed by repulsion of  eigenvalues, see for example \cite{JakMilRivRud99,Elo08,OreSmi10,ErdKnoYauYin12}.

In this article a small  step is made in this direction by studying convergence of the density of states and delocalization of eigenvectors. First we analyze the density of states.
The analogue of Wigner's result for random regular graphs was proved by McKay \cite{Kay81}: The  density of states of the adjacency matrix of a random regular graph converges in distribution to a probability measure, known as the Kesten-McKay law. 

In Theorem \ref{thm:adj} we refine this result by giving uniform bounds on the rate of convergence. This allows us to deduce local convergence of the number of eigenvalues in intervals of size comparable to $\log_{d-1}^{-1}(n)$. Our approach is based on the fact that a random regular graph coincides locally with a regular tree. This explains the rate $\log_{d-1}^{-1}(n)$ since this approximation typically works well up to distances comparable to $\log_{d-1}(n)$.
While this rate is far from the desired $1/ n$, our results are strong enough to deduce statements about convergence of the Green function and delocalization of eigenvectors.

Our findings are similar to results recently obtained by Dumitriu and Pal \cite{DumPal12} and by Tran, Vu, and Wang \cite{TraVuWan13}. They also study spectral properties of random regular graphs and consider the case where the degree $d$ is not fixed but tends to infinity together with the number of vertices $n$. The main results also include delocalization of eigenvectors and local convergence of the number of eigenvalues. Again the size of the allowed intervals is comparable to $\log_{d_n-1}^{-1}(n)$.
Let us remark that our methods -- even though they are tailored for regular graphs with fixed degree --  also extend to graphs with growing degree. In Theorem~\ref{thm:grow} we recover results from \cite{DumPal12} and \cite{TraVuWan13} and we slightly improve them in a certain sense that is made precise below.

After this prelude, the main purpose of this article is to derive similar results for random Schr\"odinger operators and to deduce delocalization of eigenvectors.
Motivated by recent physical and numerical considerations about eigenvalue statistics \cite{DelSca13,BirReiTar13} we study the Anderson model of random Schr\"odinger operators on random regular graphs: A random Schr\"odinger operator consists of the adjacency matrix perturbed by a random potential,  that means one adds independent identically distributed entries on the diagonal of the matrix. Again we explore the density of states by comparing with the Anderson model on the infinite regular tree. The strength of our approach lies in the fact that it depends only on local properties of the graph, therefore it extends to general local operators such as random Schr\"odinger operators.

The Anderson model on the tree is one of the most studied models of random Schr\"odinger operators, see \cite{War12} for an overview of results and references. It is a natural question in what way spectral properties on the infinite tree extend to the  corresponding finite-volume operator on a random regular graph. However, for Schr\"odinger operators the analysis of spectrum  and eigenvectors is more challenging since the corresponding spectral measure on the infinite tree is not purely absolutely continuous but it also contains a pure-point component.

This should influence the behavior of eigenvectors of the finite-volume operator. In spectral regimes that correspond to pure-point spectrum of the infinite-volume operator one expects to find exponential localization while eigenvectors with eigenvalues within the absolutely continuous spectrum are expected to be delocalized. In turn, properties of eigenvectors are conjectured to determine whether the local eigenvalue statistics is Poisson or governed by level repulsion \cite{DelSca13,BirReiTar13,War12} . 

Again we are able to make a small step in this direction. In Theorem \ref{thm:esd} we show that the mean density of states of a random Schr\"odinger operator on a random regular graph converges to the density of states on the infinite tree. We give bounds on the rate of convergence and show that the mean number of eigenvalues in intervals of size comparable to $\log_{d-1}^{-1}(n)$ converges.
Finally, in Theorem \ref{thm:deloc} we apply these results and combine them with recent results about random Schr\"odinger operators on the infinite tree \cite{AizWar12} to prove that typically eigenvectors on a random regular graph with eigenvalues within the absolutely continuous spectrum of the infinite-volume operator are not localized. 

The article is organized as follows. In the next section we explain the relevant notation concerning random regular graphs and spectral measures. Then we give the precise statements of the main results about convergence of the density of states.

In Section \ref{sec:detest} we provide the main tools: The fact that a random regular graph coincides locally with a tree implies that low moments of the spectral measure on the graph agree with the respective moments on the tree. Therefore also the expectation of polynomials of low degree (typically up to $\log_{d-1}(n)$) is the same.  Thus to obtain uniform estimates on the rate of convergence of the density of states one needs to approximate the Heaviside function by polynomials. As was noticed by Chebyshev, Markov, and Stieltjes orthogonal polynomials are well suited for this purpose.  In Theorem~\ref{thm:detest} we use this approach to prove a deterministic estimate for the density of states valid for all regular graphs.

In Section \ref{sec:cycles} we study the distribution of cycles in a random regular graph and we show that it is justified to approximate a random regular graph locally by a tree. We use that to complete the proof of the results from Section \ref{sec:results}.

In the final section we apply the developed techniques to prove convergence of the Green function and to deduce results about delocalization of eigenvectors. 
 

\section{Notation and results}
\label{sec:results}

Consider the set $\mathcal G_{n,d}$ of simple regular graphs with $n$ vertices and degree $d \geq 3$. Let $\mathcal P_{n,d}$ denote the uniform probability measure on this set and let $\mathcal E_{n,d}$ denote the expectation with respect to $\mathcal P_{n,d}$. We study this ensemble as $n$ tends to infinity. We say that an event happens asymptotically almost surely if the probability $\mathcal P_{n,d}$ of the event tends to one as $n$ tends to infinity.  

We assume that the degree $d$ is fixed unless stated otherwise. However, the methods that are employed here are not limited to this case. In particular in Theorem \ref{thm:grow} we consider the case where $d = d_n$ tends to infinity with $n$.

\subsection{The adjacency matrix}
\label{ssec:adj}

First we study the spectrum of  the adjacency matrix $A_n$ of a random regular graph $G_n \in \mathcal G_{n,d}$. The adjacency matrix is the self-adjoint operator on the Hilbert space $\ell^2(G_n)$  defined by
$$
\lk A_n \phi \rk(x) = \sum_{y \in G_n \, : \, d(x,y) = 1} \phi (y) \, , \qquad \phi \in \ell^2(G_n) \, , \qquad x \in G_n \, .
$$
Here the distance $d(x,y)$ of two vertices $x, y \in G_n$  is the length of the shortest path connecting $x$ and $y$. The adjacency matrix is the discrete Laplace operator on $G_n$ with the diagonal terms removed.

Let $(\lambda_j)_{j=1}^n$ denote the eigenvalues of $A_n$ and let $(\varphi_j)_{j=1}^n$ be the corresponding $\ell^2(G_n)$-normalized eigenvectors. If an eigenvalue has multiplicity higher than one we repeat the value according to its multiplicity and we choose the eigenvectors as an orthonormal basis of the eigenspace. Since $G_n$ has $n$ vertices, $A_n$ is a symmetric  $n$ by $n$ matrix and we obtain $n$ eigenvalues and $n$ eigenvectors. To study the distribution of the eigenvalues we introduce the following spectral measures.

For a vertex $x \in G_n$ we write $\delta_x \in \ell^2(G_n)$ for its characteristic function: $\delta_x(x) = 1$ and $\delta_x(y) = 0$ for $y \neq x$. Also, for a set $I \subset \R$ let $\chi_I$ denote its characteristic function: $\chi_I(t) = 1$ for $t \in I$ and $\chi_I(t) = 0$ for $t \notin I$. We write $|I|$ for the length of the set.
The local spectral measure $\mu_{n,x}$ with respect to a vertex $x \in G_n$ is given by
\begin{equation}
\label{adjmeasure}
\mu_{n,x}(I) = \lk \delta_x, \chi_I \lk A_n \rk \delta_x \rk_{\ell^2(G_n)} = \sum_{\lambda_j \in I} |\varphi_j(x)|^2 \, .
\end{equation}
With $\mathcal N_I(G_n)$ we denote the counting measure that counts the number of eigenvalues $\lambda_j$ in the set $I$,
$$
\mathcal N_I(G_n) = \sum_{j=1}^n \chi_I(\lambda_j) = \tr \left[ \chi_I \lk  A_n \rk  \right]
$$
and we remark the identity
\begin{equation}
\label{adjmeasures}
\mathcal N_I(G_n) = \sum_{x \in G_n} \mu_{n,x}(I) \, .
\end{equation}

Our main tool in the analysis of the spectral distribution of $A_n$ is the fact that a typical regular graph locally resembles a tree in the sense that it contains large regions without cycles. (A cycle is a closed path without repetitions of vertices and edges other than the starting and ending vertex.) So we also consider the infinite regular tree $\mathcal T_d$ of degree $d$. On the tree the adjacency matrix $A_{\mathcal T_d}$ is again defined by
\[
\lk A_{\mathcal T_d} \phi \rk(x) = \sum_{y \in \mathcal T_d \, : \, d(x,y) = 1} \phi (y) \, , \qquad \phi \in \ell^2(\mathcal T_d) \, , \qquad x \in \mathcal T_d \, .
\]
Since the degree $d$ is finite $A_{\mathcal T_d}$ is a bounded and self-adjoint operator with domain $\ell(\mathcal T_d)$. In analogy with the local spectral measure \eqref{adjmeasure} we define the local density of states measure
\begin{equation}
\label{treeadjmeasure}
\sigma_0(I) = \lk \delta_x, \chi_I \lk A_{\mathcal T_d} \rk \delta_x \rk_{l^2(\mathcal T_d)}  
\end{equation}
which is independent of $x \in \mathcal T_d$.
On the tree this can be calculated explicitly and is given by the Kesten-McKay measure:
\begin{equation}
\label{kesten}
\sigma_0(d\lambda) =\frac{d}{2\pi}  \frac{\sqrt{4(d-1)-\lambda^2}}{d^2-\lambda^2}  \chi_{(-2\sqrt{d-1}, 2\sqrt{d-1})} (\lambda) d\lambda \, .
\end{equation}

The fact that a random regular graph is locally identical to a tree has already been used by McKay to determine the limiting distribution of eigenvalues of the adjacency matrix. Assume that for each fixed $k \geq 3$ the number of cycles of length $k$ in $G_n$ is of order $o(n)$ as $n$ tends to infinity. Then it is shown in \cite{Kay81} that the measure $\mathcal N_{(\cdot)}(G_n) /n$ converges in distribution to $\sigma_0$. This local convergence was generalized in \cite{Sal11} to the much richer class of self-adjoint graphs that includes trees and, in particular, regular trees. Since we aim at comparing with results on regular trees we restrict ourselves here to local convergence of random regular graphs to regular trees.

We refine the result of McKay by giving an estimate on the rate of convergence.
Note that the measure $\sigma_0$ is supported in the interval $[-2\sqrt {d-1}, 2 \sqrt{d-1}]$ and that its density is bounded by $\gamma_d$, where
$$
\gamma_d = \frac{d}{4\pi} \frac{1}{\sqrt{d^2-4(d-1)}} \ \textnormal{if} \ d \leq 6 \quad \textnormal{and} \quad \gamma_d =  \frac{\sqrt{d-1}}{d \pi} \ \textnormal{if} \ d \geq 7 \, .
$$

\begin{theorem}
\label{thm:adj}
The local density of states of the adjacency matrix of a random regular graph satisfies the estimate
$$
\sup_{t \in \R} \left[ \frac 1n \sum_{x \in G_n} \left| \mu_{n,x}((-\infty,t]) - \sigma_0 ((-\infty,t])  \right| \right] \leq  C \gamma_d    \sqrt {d-1} \log^{-1}_{d-1}(n) 
$$
asymptotically almost surely for any constant $C > 8\pi$.

In particular, the estimate
\begin{equation}
\label{eq:fixedn}
\left| \frac 1n \mathcal N_I(G_n) - \sigma_0(I) \right| \leq \delta  |I|
\end{equation}
holds asymptotically almost surely, for all $\delta > 0$ and all intervals $I \subset \R$ satisfying 
$$
|I| \geq   \frac { 2C \gamma_d  \sqrt {d-1}}{\delta} \log_{d-1}^{-1}(n) \, .
$$

\end{theorem}

\begin{rem}
As mentioned in the introduction this result can be seen as a generalization of recent results found in \cite{DumPal12,TraVuWan13}, where the spectral distribution of the adjacency matrix is analyzed on random regular graphs with degree $d = d_n$ that tends to infinity as $n \to \infty$. In fact from \cite[Lemma 10]{DumPal12} one can also derive a result for fixed degree $d$, namely that \eqref{eq:fixedn} holds if the size of the interval $I$ is larger than $\log (\log (n)) \log^{-1}(n)$.
\end{rem}

To emphasize that our approach can be applied in various situations let us  now briefly consider the case where the degree $d = d_n$ depends on the number of vertices and tends to infinity as $n \to \infty$, as studied in \cite{DumPal12,TraVuWan13}. To confine the spectrum to a finite interval one considers the rescaled operator 
$$
\tilde A_n = \frac 1{2\sqrt{d_n-1}} A_n 
$$
on $\ell^2(G_n)$ and the correspondingly rescaled operator $\tilde A_{\mathcal T_{d_n}}$ on $\ell^2(\mathcal T_{d_n})$.
Then the local density of states measure $\tilde \sigma_0(I) = ( \delta_x, \chi_I ( \tilde A_{\mathcal T_{d_n}} ) \delta_x )_{\ell^2(\mathcal T_{d_n})}$ is given by
\begin{equation}
\label{eq:rescaledmes}
\tilde \sigma_0(d\lambda) = \frac 2 \pi    \frac{d_n (d_n-1)}{d_n^2-4(d_n-1)\lambda^2}  \sqrt{1 - \lambda^2} \chi_{(-1,1)} (\lambda) d\lambda \, .
\end{equation}
As $n$ tends to infinity this measure converges to the semicircle measure
\begin{equation}
\label{sc}
\sigma_{\s}(d\lambda) = \frac 2 \pi \sqrt{1 - \lambda^2} \chi_{(-1,1)}(\lambda) d\lambda \, .
\end{equation}
In the same way as above we define the local spectral measure $\tilde \mu_{n,x}(I) = ( \delta_x, \chi_I ( \tilde A_n ) \delta_x )_{\ell^2(G_n)}$
and the counting measure
\begin{equation}
\label{growmeasures}
\tilde{ \mathcal  N}_I(G_n) = \tr \left[ \chi_I (  \tilde A_n ) \right] = \sum_{x \in G_n} \tilde \mu_{n,x}(I)
\end{equation}
and we obtain the following local semicircle law. 

\begin{theorem}
\label{thm:grow}
Let $d_n \to \infty$ as $n \to \infty$ with $d_n \leq (n/\ln(n))^{1/3}$. Then the local density of states of the operator $\tilde A_n$ satisfies the estimate
$$
\sup_{t \in \R} \left[ \frac 1n \sum_{x \in G_n} \left| \tilde \mu_{n,x}((-\infty,t]) - \sigma_\s ( (-\infty,t] )  \right| \right] \leq C  \lk  \frac{\ln (d_n-1)}{\ln(n)} + \frac 1{ d_n} \rk  
$$
asymptotically almost surely for any constant $C > 8$.

In particular, the estimate
\begin{equation}
\label{grownumber}
\left| \frac 1n \tilde{ \mathcal N}_I(G_n) - \sigma_\s(I) \right| \leq \delta  |I|
\end{equation}
holds asymptotically almost surely, for all $\delta > 0$ and all intervals $I \subset \R$ satisfying 
$$
|I| \geq   \frac{2 C}{\delta}  \lk  \frac{\ln (d_n-1)}{\ln(n)} + \frac 1{d_n} \rk \, . 
$$

\end{theorem}

\begin{rem}
This theorem is similar to recent results from \cite{DumPal12,TraVuWan13}. In \cite[Theorem 2]{DumPal12} it is shown that \eqref{grownumber} holds with probability at least $1- o(1/n)$, if $d_n = \ln^\gamma (n)$ and the size of the considered interval $I$ is comparable to $\ln^{-\beta \gamma}(n)$ for $0 < \gamma < 1$ and comparable to $\ln^{-\beta}(n)$ for $\gamma \geq 1$ both with $\beta < 1$. Theorem \ref{thm:grow} above improves this in the sense that one can consider slightly smaller intervals $I$.

 In \cite[Theorem 1.6]{TraVuWan13} a similar estimate is shown to hold with probability at least $1-O(\exp(-cn\sqrt {d_n} \ln (d_n)))$ if the size of the considered interval $I$ is comparable to $\ln^{1/5}(d_n)$ $d_n^{-1/10}$. As far as the size of the interval is concerned Theorem \ref{thm:grow} above gives better bounds if $d_n$ is less than $\ln^{10}(n)$. For $d_n$ larger than that the result in \cite{TraVuWan13} is stronger. 
\end{rem}


\subsection{Random Schr\"odinger operators}
\label{ssec:rs}

Now we study the distribution of eigenvalues of a random Schr\"odinger operator on a random regular graph $G_n \in \mathcal G_{n,d}$ with fixed degree $d \geq 3$. We define the operator
$$
H_n(V) = A_n + V
$$
on $\ell^2(G_n)$.  
The operator $V$ denotes a random potential, a multiplication operator,
$$
\lk V \phi \rk (x) = \omega_x \phi(x) \, , \quad \phi \in \ell^2(G_n) \, ,\quad x \in G_n \, ,
$$ 
where $\lk \omega_x \rk_{x \in G_n}$ stands for a collection of independent identically distributed real random variables with density $\rho$. We assume that $\rho$ has bounded support such that $\supp \rho \subset (-\rho_0, \rho_0)$ for some $0 <\rho_0 < \infty$ and we write $\| \rho \|_\infty = \sup_{t \in \R} \rho(t)$. 

For random Schr\"odinger operators we apply the same notation as above for eigenvalues and eigenvectors. Here the eigenvalues $(\lambda_j)_{j=1}^n$  and eigenvectors $(\varphi_j)_{j=1}^n$ are random objects even for a fixed graph $G_n$ since they depend on the potential $V$. So for $x \in G_n$ we consider the random spectral measure
\begin{equation}
\label{eq:randmes}
\mu_{n,x}(I;V) = \lk \delta_x, \chi_I \lk H_n(V) \rk \delta_x \rk_{\ell^2(G_n)} = \sum_{\lambda_j \in I} |\varphi_j(x)|^2
\end{equation}
that corresponds to \eqref{adjmeasure}. As in \eqref{adjmeasures}, with $\mathcal N_I(G_n;V)$ we denote the random variable that counts the number of eigenvalues $\lambda_j$ in the set $I$,
\begin{equation}
\label{romeasures}
\mathcal N_I(G_n;V) = \sum_{j=1}^n \chi_I(\lambda_j) =  \tr \left[ \chi_I \lk H_n(V) \rk \right] = \sum_{x \in G_n} \mu_{n,x}(I;V) \, .
\end{equation}

In the same way we define the corresponding objects on the tree $\mathcal T_d$: First the operator
\[
H_{\mathcal T_d}(V) = A_{\mathcal T_d} + V
\]
on $\ell^2(\mathcal T_d)$. Here $V = (\omega_x)_{x \in \mathcal T_d}$  denotes again a collection of independent identically distributed real random variables with density $\rho$. We write $\p$ and $\E$ for the probability and expectation with respect to the distribution of $(\omega_x)_{x \in \mathcal T_d}$. We recall that the density of the random potential has finite support. Hence, $H_{\mathcal T_d}$ is a bounded self-adjoint operator with domain $\ell^2(\mathcal T_d)$. 

For $x \in \mathcal T_d$ we define the local density of states measure
$$
\mu_{\mathcal T_d,x}(I;V) = \lk \delta_x, \chi_I \lk H_{\mathcal T_d}(V) \rk \delta_x \rk_{\ell^2(\mathcal T_d)} \, .
$$
This measure depends on the potential $V$ and thus on the vertex $x \in \mathcal T_d$. To obtain an invariant measure we take the expectation with respect to the random potential and set
\begin{equation}
\label{treedos}
\sigma_\rho(I) = \E  \left[ \mu_{\mathcal T_d,x}(I;V) \right] \, .
\end{equation}
By translation invariance of the operator $H_{\mathcal T_d}(V)$ the measure $\sigma_\rho$ is independent of $x \in \mathcal T_d$ and thus depends only on the density $\rho$. We refer to Appendix \ref{ap:wegner}, where we state selected properties of $\mu_{\mathcal T_d,x}$ and  $\sigma_\rho$.

To identify the random potential on the graph $G_n$ with the random potential on the  tree $\mathcal T_d$ we consider the tree as the universal cover of the graph. To construct the universal cover one starts with an arbitrary vertex $o \in G_n$ and the set of non-backtracking walks in $G_n$ that start at $o$. This is the set of finite sequences $(x_j)_{j=1}^k$ such that $x_1 = o$, $x_j$ is adjacent to $x_{j+1}$ in $G_n$ for $j = 1, \dots, k-1$, and $x_{j-1} \neq x_{j+1}$ for $j = 2, \dots, k-1$. Two such walks are said to be adjacent if their lengths differ by one and if they agree except for the last vertex of the longer walk. The set of non-backtracking walks in $G_n$ starting at $o$ with this notion of adjacency  is isomorphic to the tree $\mathcal T_d$ and forms the universal cover of the graph $G_n$. 

The graph $G_n$ induces an equivalence relation on its universal cover and thus on the tree: Two vertices of $\mathcal T_d$ are equivalent if the corresponding non-backtracking walks in $G_n$ have the same endpoints. Hence, the graph $G_n$ can be recovered from the universal cover as the set of equivalence classes. This induces a map
\[
\iota : \mathcal T_d \to G_n
\]
which is onto and preserves the  local geometry in the sense explained in the following.

For a vertex $x \in G_n$ and $k \in \N$ let $B_k(x)  = \{y \in G_n \, : \, d(x,y) \leq k \}$ denote the $k$-neighborhood of $x$, including all edges of $G_n$ that are incident with at least one vertex $y$ with $d(x,y) < k$. If $B_k(x)$ is acyclic then $B_k(x)$ is a finite tree of depth $k$. To compare the graph $G_n$ locally to the tree  we define
\begin{equation}
\label{eq:maxr}
R(x) = \max \{k \in \N \, : \, B_k(x) \ \textnormal{is acyclic} \} \, .
\end{equation}
(Since $G_n$ does not contain double edges we have $R(x) \geq 1$ for all $x \in G_n$.)
Let $\hat x \in \mathcal T_d$ be an arbitrary vertex from the preimage of $x$ under the map $\iota$ and let 
\begin{equation}
\label{eq:iotamap}
\iota_{\hat x} \, : \, B_{R(x)}(\hat x) \subset \mathcal T_d \, \to \, B_{R(x)}(x) \subset G_n
\end{equation}
be the restriction of $\iota$ to the neighborhood $B_{R(x)}(\hat x)=\{ \hat y \in \mathcal T_d \, : \, d(\hat x, \hat y) \leq R(x) \}$. By \eqref{eq:maxr} the map $\iota_{\hat x}$ is an isomorphism from $B_{R(x)}(\hat x)$ to $B_{R(x)}(x)$.

Given a realization of the random potential $V = (\omega_{\hat x})_{\hat x \in \mathcal T_d}$ on the tree, this map generates a realization of the potential on  the graph: For $x \in G_n$ we choose  $\hat x \in \mathcal T_d$ as above and for $y \in B_{R(x)}(x)$ we set $\omega_y = \omega_{\hat y}$, where $\hat y \in B_{R(x)}(\hat x)$ is the unique preimage of $y$ under $\iota_{\hat x}$. For $y \notin B_{R(x)}(x)$ we set $\omega_y = \omega_{\hat y}$, where $\hat y \in \mathcal T_d$ is an arbitrary vertex from the preimage of $y$ under $\iota$.  This procedure yields  a collection of independent and identically distributed random variables $(\omega_y)_{y \in G_n}$ with density $\rho$. In fact, this realization of the random potential depends on $x \in G_n$ and on the choice of preimages. However, by independence of the random variables $(\omega_y)_{y \in \mathcal T_d}$ this dependence disappears after taking expectations. Thus we denote the resulting random potential on  $G_n$ again by $V$.

With this construction the local geometry and the random potential in $B_{R(x)}(x) \subset G_n$ and $B_{R(x)}(\hat x) \subset \mathcal T_d$ coincide. By induction in $k \in \N$, it is easy to see that $(\delta_x, H_n(V)^k \delta_x)_{\ell^2(G_n)}$ and $(\delta_{\hat x}, H_{\mathcal T_d}(V)^k \delta_{\hat x})_{\ell^2(\mathcal T_d)}$ only depend on the local geometry and on the random potential in $B_m(x)$ and $B_m(\hat x)$ respectively, where $m = k/2$ for $k$ even and $m = (k+1)/2$ for $k$ odd.  In particular, it follows that 
\begin{equation}
\label{localid}
(\delta_x, H_n(V)^k \delta_x)_{\ell^2(G_n)} = (\delta_{\hat x}, H_{\mathcal T_d}(V)^k \delta_{\hat x})_{\ell^2(\mathcal T_d)}
\end{equation}
for all $k = 0, 1, \dots, 2R(x)$. This identity is a key ingredient in the proof of the following result.

\begin{theorem}
\label{thm:esd}
Assume that the density $\rho$ of the random potential $V$ satisfies $\| \rho \|_\infty < \infty$ and $\supp  (\rho) \subset (-\rho_0,\rho_0)$ with $0 < \rho_0 < \infty$. 
Then the local density of states of the operator $H_n(V)$ on a random regular graph satisfies the estimate
$$
\sup_{t \in \R} \left[ \frac 1n \sum_{x \in G_n} \E \left| \mu_{n,x}((-\infty,t];V) - \mu_{\mathcal T_d,\hat x} ((-\infty,t];V)  \right| \right] \leq C \| \rho \|_\infty \lk 2 \sqrt {d-1} + \rho_0 \rk \log_{d-1}^{-1}(n) 
$$
asymptotically almost surely for any constant $C > 4 \pi$.

In particular, the estimate
$$
\left| \frac 1n \E \left[\mathcal N_I(G_n;V) \right] - \sigma_\rho(I) \right| \leq \delta  |I|
$$
holds asymptotically almost surely, for all $\delta > 0$ and all intervals $I \subset \R$ satisfying 
$$
|I| \geq   \frac {2C  \| \rho \|_\infty \lk 2 \sqrt {d-1} + \rho_0 \rk}{\delta} \log_{d-1}^{-1}(n) \, .
$$
\end{theorem}

In Section \ref{sec:eigenfunctions} we apply the developed methods to deduce estimates for the Green function. We establish convergence of the imaginary part of the Green function on a random regular graph to the respective quantity on the infinite regular tree, see Corollary~\ref{cor:green}. Based on these bounds we prove statements about delocalization of eigenvectors. In particular, in  Theorem~\ref{thm:deloc} we show that eigenvectors of the operator $H_n(V)$ with eigenvalues within the absolutely continuous spectrum of the infinite-volume operator $H_{\mathcal T_d}(V)$ are not localized. Since these results require more notation and since they ask for some discussion we state these results in Section~\ref{sec:eigenfunctions}.


\section{A deterministic estimate for the integrated density of states} 
\label{sec:detest}

In this section we fix a graph $G_n \in \mathcal{G}_{n,d}$ and a vertex $x \in G_n$. We derive an  estimate for the difference between the spectral measure on $G_n$ and the respective measure on the tree $\mathcal T_d$ that depends on the local geometry of the graph.

Recall the definition of $R(x)$ from \eqref{eq:maxr} and set $R(x)^* = R(x)$ for $R(x)$ odd and $R(x)^* = R(x)-1$ for $R(x)$ even.
As in Theorem~\ref{thm:esd}  we write $\hat x \in \mathcal T_d$ for a  vertex from the preimage of $x \in G_n$ under the universal cover.
  
\begin{theorem}
\label{thm:detest}
For all $G_n \in \mathcal G_{n,d}$ and all $x \in G_n$ the local density of states of $A_n$ satisfies 
$$
\sup_{t\in \R} \left[ \left| \mu_{n,x}\lk (-\infty, t] \rk  - \sigma_0 \lk (-\infty,t] \rk \right| \right] \leq  4\pi \gamma_d \sqrt {d-1} \frac{1}{R(x)^*} \, .
$$
Moreover, under the conditions of Theorem~\ref{thm:esd}, the local density of states of $H_n(V)$ satisfies 
$$
\sup_{t\in \R} \left[ \E \left|  \mu_{n,x}\lk (-\infty,t];V \rk  - \mu_{\mathcal T_d, \hat x}\lk (-\infty,t];V \rk  \right| \right] \leq 2\pi  \| \rho \|_\infty \lk 2 \sqrt {d-1} + \rho_0 \rk   \frac{1}{R(x)^*} 
$$
for all $G_n \in \mathcal G_{n,d}$ and all $x \in G_n$.
\end{theorem}

\begin{rem}
For the adjacency matrix $A_n$ there is a variant of this result. In \cite{Sod07} it is shown that there is a constant $C>0$ such that, for all $m \in \N$,
$$
\sup_{t\in \R} \left[ \left| \mu_{n,x}\lk (-\infty, t] \rk  - \sigma_0 \lk (-\infty,t] \rk \right| \right] \leq C \lk \frac 1m + m^6 \lk \sum_{k = 1}^{2m-2} W_k(x,G_n)^2 \rk^{1/2} \rk \, ,
$$
where $W_k(x,G_n)$ is related to the number of closed non-backtracking walks of length $k$ in $G_n$ that start at $x$. For $m < R(x)$ the second summand is zero. More generally, $W_k(x,G_n)$ can be estimated in terms of the number of cycles in $G_n$ and the resulting bounds are similar to Theorem \ref{thm:adj}.
\end{rem}

In the remainder of this section we prove Theorem \ref{thm:detest}. The proof is based on a general estimate for measures on the real line that we give in the next subsection. In Subsection~\ref{ssec:detproof} we show how the theorem can be deduced from Proposition \ref{pro:gen}.


\subsection{A general estimate}

The following result is related to the classical moment problem and the Chebyshev-Markov-Stieltjes inequality; it is based on an approximation of the Heaviside function by orthogonal polynomials. We refer to the books \cite{Akh65,Sze75,KreNud77} for background information regarding this approach.

\begin{proposition}
\label{pro:gen}
Let $\sigma$ be a measure on the real line with bounded density $w$ and with support in the finite interval $(-w_0,w_0)$.
Let $N \in \N$ and assume that $\mu$ is another measure on the real line such that, for all $k= 0, \dots, 2N$,
\begin{equation}
\label{eq:momid}
\int_\R \lambda^k \sigma(d\lambda) = \int_\R \lambda^k \mu(d\lambda) \, .
\end{equation}
Then the estimate 
$$
\sup_{t\in \R} \left| \sigma((-\infty,t]) - \mu((-\infty,t]) \right| \leq \frac {2\pi} {N^*} \| w \|_\infty w_0 
$$
holds with $N^* = N$ for $N$ odd and with $N^* = N-1$ for $N$ even.
\end{proposition}

\begin{proof}
First we note that for $t \leq - w_0$ we have
$$
\left| \sigma((-\infty,t]) - \mu((-\infty,t]) \right|  \leq \left| \mu((-\infty,-\omega_0]) \right| =  \left| \sigma((-\infty,-w_0]) - \mu((-\infty,-w_0]) \right| 
$$
and for $t \geq w_0$ (by \eqref{eq:momid} with $k=0$),
$$
\left| \sigma((-\infty,t]) - \mu((-\infty,t]) \right|  =  \sigma((-\infty,\omega_0]) - \mu((-\infty,t])  \leq \left| \sigma((-\infty,w_0]) - \mu((-\infty,w_0]) \right|  \, .
$$
Hence, we can assume $t \in (-w_0,w_0)$.

For $n \in \N_0$, let $P_n$ denote the orthonormal polynomial of degree $n$ with respect to the measure $\sigma$. We claim that 
\begin{equation}
\label{cmsbound}
\left| \sigma \lk (-\infty,t] \rk - \mu \lk (-\infty,t] \rk \right| \leq  \frac 1{\sum_{n=0}^N P_n(t)^2 } \, .
\end{equation}
Combining this bound with Lemma \ref{lem:christoffel} below proves the proposition.

To establish \eqref{cmsbound} let us first assume that $t$ is a zero of the polynomial $P_N$. We  remark that these zeros are real and simple \cite{Akh65} and we denote them by $\lambda_1 < \lambda_2 < \dots < \lambda_N$. Then this assumption means that $t = \lambda_j$ for an index $j \in \{1, \dots, N \}$.  
We construct a polynomial $R_{2N-2}$ of degree $2N-2$ that satisfies
\[
R_{2N-2}(\lambda_1) =  \cdots = R_{2N-2}(\lambda_j) = 1 \ , \qquad R_{2N-2}(\lambda_{j+1}) =  \cdots=  R_{2N-2}(\lambda_N) = 0 \, ,
\]
and 
\[
R_{2N-2}'(\lambda_i) = 0 
\]
for all $i \neq j$.
These $2N-1$ assumptions determine the polynomial $R_{2N-2}$ uniquely and we see that $R_{2N-2}(\lambda) \geq \chi_{(-\infty,\lambda_j]}(\lambda)$ for all $\lambda \in \R$. In the same way we construct a polynomial $Q_{2N-2}$ of  degree $2N-2$ by changing only the condition at $\lambda_j$ to $Q_{2N-2}(\lambda_j) = 0$. Then we get $Q_{2N-2} (\lambda) \leq \chi_{(-\infty,\lambda_j]}(\lambda)$ for all $\lambda \in \R$.
Hence, we can estimate
\begin{align}
\nonumber
\sigma\lk (-\infty, \lambda_j] \rk &= \int_\R \chi_{(-\infty,\lambda_j]} \sigma(d\lambda) \leq \int_\R R_{2N-2}(\lambda) \sigma(d\lambda) \, , \\
\label{upperandlower}
\mu \lk (-\infty, \lambda_j] \rk &= \int_\R \chi_{(-\infty,\lambda_j]} \mu(d\lambda) \geq \int_\R Q_{2N-2}(\lambda) \mu(d\lambda)
\end{align}
and
\begin{equation}
\label{beforequad}
\sigma\lk (-\infty, \lambda_j] \rk - \mu \lk (-\infty, \lambda_j] \rk  \leq \int_\R R_{2N-2}(\lambda) \sigma(d\lambda) - \int_\R Q_{2N-2}(\lambda) \mu(d\lambda) \, .
\end{equation}

To bound the right-hand side we invoke the Gaussian quadrature formula from Lemma~\ref{lem:quadrature} in Appendix \ref{ap:quadrature}: With $M = N-1$ and $s = 0$ we find
$$
\int_\R R_{2N-2}(\lambda) \sigma(d\lambda) = \sum_{k=1}^N \frac{R_{2N-2}(\lambda_k)}{\sum_{n=0}^{N-1} P_n(\lambda_k)^2} =  \sum_{k=1}^j \frac 1{\sum_{n=0}^{N-1} P_n(\lambda_k)^2} \, .
$$
By assumption, the first $2N$ moments of the measures $\sigma$ and $\mu$ agree, so we also get 
$$
\int_\R Q_{2N-2}(\lambda) \mu(d\lambda) = \int_\R Q_{2N-2}(\lambda) \sigma(d\lambda) = \sum_{k=1}^N \frac{Q_{2N-2}(\lambda_k)}{\sum_{n=0}^{N-1} P_n(\lambda_k)^2} =  \sum_{k=1}^{j-1} \frac 1{\sum_{n=0}^{N-1} P_n(\lambda_k)^2} \, .
$$
Combining these identities with the estimate \eqref{beforequad} yields the upper bound
$$
\sigma\lk (-\infty, \lambda_j] \rk - \mu \lk (-\infty, \lambda_j] \rk \leq  \frac 1{\sum_{n=0}^{N-1} P_n(\lambda_j)^2} =  \frac 1{\sum_{n=0}^N P_n(\lambda_j)^2} \, ,
$$
where we used the assumption $P_N(\lambda_j) = 0$ in the last step. The lower bound is proved in the same way by exchanging the roles of $\sigma$ and $\mu$ in \eqref{upperandlower} and \eqref{beforequad}. This proves \eqref{cmsbound} if $t$ is a zero of $P_N$.

It remains to prove \eqref{cmsbound} if $t$ is not a zero of $P_N$. For $s \in \R$ we define a polynomial of degree $N+1$ by
$$
\hat P_{N+1}(\lambda) = P_{N+1}(\lambda) + s P_N(\lambda) \, .
$$
Since $P_N(t) \neq 0$ we can choose $s$ in such a way that $\hat P_{N+1}(t) = 0$.
Thus we can argue in the same way as before, with $N$ replaced by $N+1$.

This establishes \eqref{cmsbound} and completes the proof.
\end{proof}

The proof of Proposition \ref{pro:gen} is based on the following estimate of the so-called Christoffel numbers $(\sum_{n=0}^N P_n(t)^2)^{-1}$.

\begin{lemma}
\label{lem:christoffel}
Assume that $\sigma$ is a measure on the real line with bounded density $w$ and with support in the finite interval $(-w_0, w_0)$. Let $P_n$, $n \in \N_0$, denote the orthonormal polynomial of degree $n$ with respect to $\sigma$.

Then for all $N \in \N$ and all $t \in (-w_0,w_0)$ the estimate
$$
\frac 1{\sum_{n=0}^N P_n(t)^2} \leq  \frac{2 \pi w_0 \|w\|_\infty}{N^*}
$$
holds with $N^* = N$ for $N$ odd and with $N^* = N-1$ for $N$ even.
\end{lemma}

\begin{proof}
We fix $N \in \N$ and $t \in (-w_0,w_0)$.  
Below we construct a polynomial $S_{2N-2}^{(t)}$ of degree less or equal than $2N-2$ with the properties
\begin{equation}
\label{fejer1}
S_{2N-2}^{(t)}(\lambda) \geq 0 
\end{equation}
for all $\lambda \in \R$,
\begin{equation}
\label{fejer2}
S_{2N-2}^{(t)}(t) = 1 \, , 
\end{equation}
and 
\begin{equation}
\label{fejer3}
\int_{-w_0}^{w_0} S_{2N-2}^{(t)}(\lambda) d\lambda  \leq \frac{2 \pi w_0 }{N^*}  \, .
\end{equation}

To  estimate $\sum_{n=0}^N P_n(t)^2$ in terms of this polynomial let us first assume that $t$ is a zero of $P_N$. Let $\lambda_1 < \lambda_2 < \dots < \lambda_N$ denote the zeros of $P_N$ and assume that $t = \lambda_j$ for an index $j \in \{1, \dots, N\}$. 
Then the assumption $P_N(\lambda_j)=0$ together with \eqref{fejer1} and \eqref{fejer2} implies
$$
\frac{1}{\sum_{n=0}^N P_n(t)^2} =  \frac{1}{\sum_{n=0}^{N-1} P_n(\lambda_j)^2} = \frac{S^{(t)}_{2N-2}(\lambda_j)}{\sum_{n=0}^{N-1} P_n(\lambda_j)^2} \leq \sum_{k=1}^N \frac{S^{(t)}_{2N-2}(\lambda_k)}{\sum_{n=0}^{N-1} P_n(\lambda_k)^2} \, .
$$
We combine this estimate with the quadrature formula from Lemma \ref{lem:quadrature} and get
\begin{equation}
\label{christoffelest}
\frac{1}{\sum_{n=0}^N P_n(t)^2} \leq \int_\R S_{2N-2}^{(t)} (\lambda) \sigma(d\lambda) \, .
\end{equation}

If $t \in (-w_0,w_0)$ is not a zero of $P_N$ we define, for $s \in \R$, a polynomial
$$
\hat P_{N+1}(\lambda) = P_{N+1}(\lambda) + s P_N(\lambda) \, .
$$
Since $P_N(t) \neq 0$ we can choose $s$ such that $\hat P_{N+1}(t) = 0$. Now we can argue similarly as above: Again let $\lambda_1 < \dots < \lambda_{N+1}$ denote the zeros of $\hat P_{N+1}$ so that $t = \lambda_j$ for an index $j \in \{1, \dots, N+1\}$.  Now we apply Lemma \ref{lem:quadrature} directly to 
$$
\frac{1}{\sum_{n=0}^N P_n(t)^2}  = \frac{S^{(t)}_{2N-2}(\lambda_j)}{\sum_{n=0}^{N} P_n(\lambda_j)^2} \leq \sum_{k=1}^{N+1} \frac{S^{(t)}_{2N-2}(\lambda_k)}{\sum_{n=0}^{N} P_n(\lambda_k)^2} 
$$ 
and we obtain \eqref{christoffelest}. We have shown that the estimate \eqref{christoffelest} is valid for all $t \in (-w_0, w_0)$.
Now we combine this estimate with \eqref{fejer3}  and arrive at
$$
\frac{1}{\sum_{n=0}^N P_n(t)^2} \leq  \int_\R S_{2N-2}^{(t)} (\lambda) \sigma(d\lambda) \leq   \| w \|_\infty \int_{-w_0}^{w_0}S_{2N-2}^{(t)} (\lambda) d\lambda \leq  \frac{2 \pi w_0 \|w\|_\infty}{N^*}
$$
which is the claimed estimate.

It remains to construct the polynomial $S^{(t)}_{2N-2}$. Let $T_m$ denote the Chebysheff polynomial of degree $m \in \N_0$ that is defined by the equation $T_m(\cos \theta) = \cos(m\theta)$. For $N \in \N$, let $n \in \N$ denote the largest integer satisfying $4n \leq 2N-2$. Then, for $x \in \R$, we set 
$$
F_{2N-2}(x) = \frac 1{2n+1} + \frac 2{(2n+1)^2} \sum_{m = 1}^{2n} \lk 2n - m + 1 \rk (-1)^m T_{2m} (x) \, .
$$
This defines a polynomial of degree $4n \leq 2N-2$.

Let us note the following relation to the Fej{\'e}r kernel.
For $x \in (-1,1)$ write $x = \cos \theta$ with $\theta \in (0,\pi)$ and calculate 
\begin{align*}
F_{2N-2}(x) &= F_{2N-2}(\cos \theta) \\
&=  \frac 1{2n+1} + \frac 2{(2n+1)^2} \sum_{m = 1}^{2n} \lk 2n - m + 1 \rk (-1)^m \cos (2m \theta  ) \\
&= \frac 12 \frac{1}{(2n+1)^2} \lk \frac{\sin^2\lk (2n+1)(\theta- \frac \pi2)\rk}{\sin^2 \lk \theta- \frac \pi 2 \rk} +  \frac{\sin^2\lk (2n+1)(\theta + \frac \pi 2)\rk}{\sin^2 \lk \theta +\frac \pi 2 \rk}  \rk \, .
\end{align*}
This identity shows that $F_{2N-2}(0) = F_{2N-2}(\cos (\pi /2)) = 1$, that $F_{2N-2}(x) \geq 0$ for $x \in (-1,1)$, and that 
$$
\int_{-1}^1 F_{2N-2}(x) dx \leq \frac \pi {2n+1} \leq \frac{\pi}{N^*} \, .
$$
To see that $F_{2N-2}(x)$ is non-negative for all $x \in \R$ note that it vanishes together with its derivative at the points $x_k = \cos\lk \pi/2 + k \pi /(2n+1) \rk$, $k = 1, 2, \dots, n$. Together with the condition $F_{2N-2}(0) = 1$ these are $2n+1$ conditions. Since $F_{2N-2}$ is by definition an even polynomial of exact degree $4n$ these conditions determine the polynomial uniquely and show that it is non-negative.  

Thus, for $t \in (-w_0,w_0)$ and $\lambda \in \R$, we set
$$
S_{2N-2}^{(t)}(\lambda) = F_{2N-2} \lk \frac{\lambda-t}{2 w_0} \rk 
$$
and the properties \eqref{fejer1}, \eqref{fejer2}, and \eqref{fejer3} follow directly from the properties of $F_{2N-2}$. This completes the proof.
\end{proof}


\subsection{Proof of Theorem \ref{thm:detest}}
\label{ssec:detproof}

Let us first consider the adjacency matrix $A_n$. For this operator the claim follows directly from Proposition \ref{pro:gen}. We only have to show that the measures $\mu_{n,x}$ and $\sigma_0$ satisfy the conditions of the proposition with $N = R(x)$: Recall that the measure $\sigma_0$ is supported in the finite interval $(-2\sqrt{d-1}, 2 \sqrt{d-1})$ and that its density is bounded by $\gamma_d$. Hence it remains to establish that the $k$-th moments of $\sigma_0$ and $\mu_{n,x}$ agree for $k = 0, \dots, 2R(x)$. 

For $x \in G_n$ we pick a vertex $\hat x \in \mathcal T_d$ from the preimage of $x$ under the universal cover and consider the map $\iota_{\hat x}$ from \eqref{eq:iotamap}. It maps the neighborhood $B_{R(x)}(\hat x) \subset \mathcal T_d$ isomorphically to $B_{R(x)}(x) \subset G_n$ and as in \eqref{localid} we find $( \delta_x, A_n^k \delta_x )_{\ell^2(G_n)} = ( \delta_{\hat x}, A_{\mathcal T_d}^k \delta_{\hat x} )_{\ell^2(\mathcal T_d)}$
for $k = 0,1,\dots, 2R(x)$. By definition of the measures $\mu_{n,x}$ and $\sigma_0$, see \eqref{adjmeasure} and \eqref{treeadjmeasure},  this implies
$$
\int_\R \lambda^k \mu_{n,x}(d\lambda) = ( \delta_x, A_n^k \delta_x )_{\ell^2(G_n)} = ( \delta_{\hat x}, A_{\mathcal T_d}^k \delta_{\hat x} )_{\ell^2(\mathcal T_d)} =  \int_\R \lambda^k \sigma_0(d\lambda) 
$$
for $k = 0,1, \dots, 2R(x)$. 
Hence, these measures satisfy the conditions of Proposition \ref{pro:gen} and the proof of the first statement is complete. 

To prove the second claim we have to argue a bit more carefully. For the random operator $H_n(V)$ the spectral measure $\mu_{\mathcal T_d,x}(I;V)$ is not necessarily absolutely continuous, hence we cannot apply Proposition \ref{pro:gen} directly. 
(We note that $\sigma_\rho$ has bounded density by the Wegner estimate \eqref{wegner}, so we can apply the proposition to the measures $\sigma_\rho$ and $\E \left[ \mu_{n,x} \right]$. This immediately yields a similar estimate for the difference of $\E \left[ \mu_{n,x} \right]$ and $\sigma_\rho$, however, we want to prove a stronger statement.)

From identity \eqref{localid} we see that 
$$
\int_\R \lambda^k \mu_{n,x}(d\lambda;V) =( \delta_x, H_n(V)^k \delta_x )_{\ell^2(G_n)} = ( \delta_{\hat x}, H_{\mathcal T_d}(V)^k \delta_{\hat x} )_{\ell^2(\mathcal T_d)} =  \int_\R \lambda^k \mu_{\mathcal T_d,\hat x}(d\lambda;V)  
$$
is valid for all $k = 0, \dots, 2R(x)$. Hence, we can apply estimate \eqref{cmsbound} to the measures $\mu_{n,x}$ and $\mu_{\mathcal T_d,\hat x}$ and obtain
\begin{equation}
\label{randomcms}
\left| \mu_{n,x}\lk (-\infty,t]; V \rk - \mu_{\mathcal T_d,\hat x}\lk (-\infty,t]; V \rk \right| \leq \frac 1{\sum_{n=0}^{R(x)}P_n(t)^2} \, ,
\end{equation}
where $P_n$ denotes the orthonormal polynomial of degree $n$ with respect to the random measure   $\mu_{\mathcal T_d,\hat x}$.

From the general property \eqref{spectralsupport} we learn the the support of this measure is almost surely contained in the interval $(-2\sqrt{d-1}-\rho_0,2\sqrt{d-1}+\rho_0)$. Hence, in the same way as in the beginning of the proof of Proposition \ref{pro:gen}  we can reduce the problem to $t \in (-2\sqrt{d-1}-\rho_0, 2\sqrt{d-1} + \rho_0)$. 
For such $t$ estimate \eqref{christoffelest} is valid and combining it with \eqref{randomcms} we find that the bound
$$
\left| \mu_{n,x}\lk (-\infty,t]; V \rk - \mu_{\mathcal T_d,\hat x}\lk (-\infty,t]; V \rk \right| \leq  \int_\R S^{(t)}_{2R(x)-2}(\lambda) \mu_{\mathcal T_d,\hat x}(d\lambda;V) 
$$
holds almost surely.
Hence, recalling definition \eqref{treedos}, the Wegner estimate \eqref{wegner}, and the bound \eqref{fejer3}, we obtain 
$$
\E \left| \mu_{n,x}\lk (-\infty,t]; V \rk - \mu_{\mathcal T_d,\hat x}\lk (-\infty,t]; V \rk \right|  \leq \int_\R S^{(t)}_{2R(x)-2}(\lambda) \sigma_\rho(d\lambda) \leq \frac{2\pi \| \rho \|_\infty ( 2 \sqrt{d-1} + \rho_0 )}{R^*(x)} \, .
$$
This proves the second claim and completes the proof of Theorem \ref{thm:detest}.


\section{Asymptotically almost sure bounds on the rate of convergence}
\label{sec:cycles}

In this section we combine the deterministic estimates of Theorem~\ref{thm:detest} with bounds on  the number of cycles in random regular graphs to prove the results from Section~{\ref{sec:results}.

\subsection{Acyclic regions in random regular graphs}
\label{ssec:acyclic}

Here we collect information about cycles in random regular graphs based on results from \cite{Kay81b,KayWorWys04}. We establish the fact that typically at most sites the graph looks locally like a tree. This is made precise in Lemma~\ref{lem:goodset} below.

For a set $A$ we write $|A|$ for the number of elements, in particular for a subset $F \subset G_n$ of a graph, $|F|$ denotes the number of vertices. 
For a graph $G_n \in \mathcal G_{n,d}$ and $k \in \N$ let $\mathcal C(k)$ denote the set of cycles of length $k$ in $G_n$. 

\begin{lemma}
\label{lem:cycles}
For $3 \leq k \leq nd/4 - 2d^2$ one has
$$
\mathcal E_{n,d} \left[ |\mathcal C(k)| \right] \leq \frac{(d-1)^k}{2k} \lk 1 +  \frac 8n \lk d+\frac{k}{2d} \rk \rk^k \, .
$$
\end{lemma}

\begin{proof}
Let $G_n$ be an arbitrary graph from $\mathcal G_{n,d}$ and let $e(G_n)$ denote the set of edges of $G_n$. We also introduce $K_n$,  the complete graph of $n$ vertices and the set $e(K_n)$ of edges of $K_n$. We consider a cycle $\mathfrak c \subset K_n$ of length $k$ and its set of edges $e(\mathfrak c)$. To estimate the expectation of $|\mathcal C(k)|$ we use the following relation to the number of cycles $\mathfrak c$ of length $k$ in $K_n$:
\begin{equation}
\label{expck}
\mathcal E_{n,d} \left[ |\mathcal C(k) | \right] = \sum_{\mathfrak c \subset K_n} \mathcal P_{n,d} \left[ e(\mathfrak c) \subset e(G_n) \right] \leq  \left| \left\{ \mathfrak c \subset K_n \right\} \right| \max_{\mathfrak c \subset K_n} \mathcal P_{n,d} \left[ e(\mathfrak c) \subset e(G_n) \right] \, .
\end{equation}

From \cite[Theorem 3]{KayWorWys04} (see also \cite[Theorem 2.10]{Kay81b}) it follows that for $k \leq nd/4 - 2d^2$ one has
$$
\mathcal P_{n,d} \left[ e(\mathfrak c) \subset e(G_n) \right] \leq \frac{d^k(d-1)^k}{2^k} \lk \frac 2{nd-4d^2-2k} \rk^k = \frac{(d-1)^k}{n^k} \lk \frac{nd}{nd-4d^2-2k} \rk^k \, .
$$
Moreover, a counting argument shows that 
$$
\left| \left\{ \mathfrak c \subset K_n \right\} \right| = \frac{n!}{2k (n-k)!} \leq \frac{n^k}{2k} \, .
$$
We combine these estimates with \eqref{expck} and get
$$
\mathcal E_{n,d} \left[ |\mathcal C(k) | \right] \leq \frac{(d-1)^k}{2k} \lk \frac{nd}{nd-4d^2-2k} \rk^k \, .
$$
A simple estimate using the fact that $k \leq nd/4 - 2d^2$ yields the claim.
\end{proof}
 
\begin{rem}
A similar argument leads to a lower bound on $\mathcal E_{n,d} \left[ |\mathcal C(k) | \right]$ with the same leading term $(d-1)^k/2k$. In this sense the estimate in Lemma \ref{lem:cycles} is sharp. 
\end{rem}
   
The following result is a simplified version of \cite[Lemma 4]{DumPal12}. It quantifies how well one can approximate a graph $G_n \in \mathcal G_{n,d}$ by a tree. 

\begin{lemma}
\label{lem:goodset}
For $k \in \N$ let  $F_n(k) \subset G_n$ denote the set of vertices $x \in G_n$ such that $B_k(x)$ is acyclic.
Then for all $\epsilon > 0$ and $k \leq n/4d-2d^2$ we have
$$
\mathcal P_{n,d} \left[ 1- \frac {|F_n(k)|}n  > \epsilon \right] \leq \frac 1{2 n \epsilon}  \frac{(d-1)^{2k+1/2}}{\sqrt{d-1}-1}  \lk 1 + \frac 8n  \lk  d+ \frac k{d}  \rk \rk^{2k}  \, .
$$
\end{lemma}

\begin{proof}
We apply the Markov inequality, namely that for all $\epsilon > 0$,
\begin{equation}
\label{markovcycle}
\mathcal P_{n,d} \left[ 1- \frac{|F_n(k)|}{n} > \epsilon \right] = \mathcal P_{n,d} \left[ |G_n \setminus F_n(k)| > n \epsilon \right] \leq \frac 1{n \epsilon} \mathcal E_{n,d} \left[ |G_n \setminus F_n(k) | \right] \, .
\end{equation}
Hence, we have to estimate $\mathcal E_{n,d} \left[ |G_n \setminus F_n(k) | \right] $.

Consider $\mathfrak c \in \mathcal C(m)$, i.e. a cycle $\mathfrak c \subset G_n$ of length $m\in \N$. For  $k \in \N$ set $N_{\mathfrak c}(k) = \{ x \in G_n \, : \, \mathfrak c \subset B_k(x) \}$. For $k < m/2$ the set $N_{\mathfrak c}(k)$ is empty. For $k \geq m/2$ it is included in the set of vertices that are at distance less or equal than $k-m/2$ from  $\mathfrak c$. Hence, we have $|N_{\mathfrak c}(k)| = 0$ for $k < m/2$ and
$$
|N_{\mathfrak c}(k)| \leq m (d-1)^{k - m/2}
$$
for $k \geq m/2$. For each vertex $x \in G_n \setminus F_n(k)$ the neighborhood $B_k(x)$ contains at least one cycle so that $G_n \setminus F_n(k) \subset \bigcup_{\mathfrak c} N_{\mathfrak c}(k)$. Hence,
$$
|G_n \setminus F_n(k) | \leq \sum_{m \geq 3} \sum_{\mathfrak c \in \mathcal C(m) } \left| N_{\mathfrak c}(k)  \right| \leq \sum_{m = 3}^{2k} m (d-1)^{k-m/2} |\mathcal C(m)| \, .
$$

We combine this estimate with Lemma \ref{lem:cycles}  and conclude that
\begin{align*}
\mathcal E_{n,d} \left[ |G_n \setminus F_n(k) | \right] &\leq \frac 12 \sum_{m=3}^{2k}  (d-1)^{k+m/2}  \lk 1 +  \frac 8n \lk d+ \frac m{2d} \rk \rk^m \\
&\leq \frac 12 \frac{(d-1)^{2k+1/2}}{\sqrt{d-1}-1}  \lk 1 + \frac 8n \lk d+\frac k{d} \rk \rk^{2k} \, .
\end{align*}
Inserting this into \eqref{markovcycle} finishes the proof.
\end{proof}


\subsection{Proof of the main results}

With Lemma \ref{lem:goodset} at hand we can deduce the results of Section~\ref{sec:results} from Theorem \ref{thm:detest}. Here we give the proofs of Theorem~\ref{thm:esd} and Theorem~\ref{thm:grow}. The proof of Theorem~\ref{thm:adj} is similar.

\begin{proof}[Proof of Theorem \ref{thm:esd}]

Recall the definition of $R(x)$ from the beginning of Section \ref{sec:detest} and the definition of $F_n(k)$ from Lemma \ref{lem:goodset}. We can always estimate $R(x) \geq 1$ and for $x \in F_n(k)$ we have $R(x) \geq k$. Thus Theorem \ref{thm:detest} implies that, for all $k \in \N$,
\begin{align}
\nonumber
&\sup_{t \in \R} \left[ \frac 1n \sum_{x\in G_n} \E \left| \mu_{n,x}((-\infty,t];V) - \mu_{\mathcal T_d,\hat x}((-\infty,t];V)  \right| \right] \\
\nonumber
&\leq 2\pi  \| \rho \|_\infty \lk 2 \sqrt {d-1} + \rho_0 \rk \frac 1n \lk \frac 1{k-1} |F_n(k)| + |G_n \setminus F_n(k) | \rk \\
\label{withbadset}
&\leq  2\pi  \| \rho \|_\infty \lk 2 \sqrt {d-1} + \rho_0 \rk  \lk \frac 1{k-1}  + \frac 1n |G_n \setminus F_n(k) | \rk \, .
\end{align}

We apply Lemma \ref{lem:goodset} to estimate the second term. For $\epsilon > 0$ let $\Omega(\epsilon,k)$ denote the event $\{ |G_n \setminus F_n(k)| \leq n \epsilon /(k-1)\}$ and note that by Lemma \ref{lem:goodset}
\begin{equation}
\label{finalprob}
\mathcal P_{n,d} \left[ \Omega(\epsilon,k) \right] \geq 1 -  \frac {k-1}{2 n \epsilon}  \frac{(d-1)^{2k+1/2}}{\sqrt{d-1}-1}  \lk 1 + \frac 8n \lk d+\frac k{d} \rk \rk^{2k}  \, .
\end{equation}
Now we choose $k = \kappa \log_{d-1}(n) +1$ with $\kappa < 1/2$. Then we see that for fixed $\epsilon > 0$ the event $\Omega(\epsilon,\kappa \log_{d-1}(n)+1)$ holds asymptotically almost surely. With this choice of $k$ estimate \eqref{withbadset} shows that the bound
\begin{align*}
\sup_{t \in \R} \left[ \frac 1n \sum_{x\in G_n} \E \left| \mu_{n,x}((-\infty,t];V) - \mu_{\mathcal T_d,\hat x}((-\infty,t];V)  \right| \right]   & \leq \frac{ (1+\epsilon) 2\pi }{k-1} \| \rho \|_\infty \lk 2 \sqrt {d-1} + \rho_0 \rk   \\
& \leq  \frac{ (1+\epsilon) 2\pi }{\kappa \log_{d-1}(n)}   \| \rho \|_\infty \lk 2 \sqrt {d-1} + \rho_0 \rk
\end{align*}
holds on the event $\Omega(\epsilon,\kappa \log_{d-1}(n)+1)$, so the bound holds asymptotically almost surely. This proves the first part of the theorem. In view of \eqref{romeasures} and \eqref{treedos} the second statement is a direct consequence of the first.
\end{proof}

\begin{proof}[Proof of Theorem \ref{thm:grow}]
Note that the measure $\tilde \sigma_0$, given in \eqref{eq:rescaledmes}, is supported in $(-1,1)$ and that its density is bounded by 
$$
\tilde \gamma_d = \frac{d \sqrt{d-1}}{2\pi} \frac{1}{\sqrt{d^2-4(d-1)}} \ \textnormal{if} \ d \leq 6 \quad \textnormal{and} \quad \tilde \gamma_d =   \frac 2 \pi \frac{d(d-1)}{d^2} \ \textnormal{if} \ d \geq 7 \, .
$$
Thus we can argue as in Section \ref{ssec:detproof} to obtain the estimate 
\begin{equation}
\label{basicgrowest}
\sup_{t \in \R}  \left[ \left| \tilde \mu_{n,x}((-\infty,t]) - \tilde \sigma_0((-\infty,t])  \right| \right] \leq \frac{2 \pi \tilde \gamma_d}{R(x)^*} \, .
\end{equation}
Hence, on the event $\Omega(\epsilon,k) = \{ |G_n \setminus F_n(k)| \leq n \epsilon /(k-1) \}$ we find that
\begin{align*}
\sup_{t \in \R} \left[ \frac 1n \sum_{x\in G_n}  \left| \tilde \mu_{n,x}((-\infty,t]) - \tilde \sigma_0((-\infty,t])  \right| \right] 
&\leq 2\pi \tilde \gamma_{d_n} \lk \frac 1{k-1} + \frac 1n |G_n \setminus F_n(k)| \rk \\
&\leq  (1+\epsilon) \tilde \gamma_{d_n}  \frac {2\pi}{k-1} \, .
\end{align*}
The probability of $\Omega(\epsilon,k)$ is bounded by \eqref{finalprob}. So we can again choose  $k = \kappa \ln (n)/\ln(d_n-1)+1$ with $\kappa < 1/2$ and we note that the assumption $d_n \leq (n/\ln(n))^{1/3}$ ensures that the condition of Lemma~\ref{lem:goodset} is satisfied. We deduce that the event $\Omega(\epsilon,\kappa \ln(n)/\ln(d_n-1)+1)$ holds, for all $\epsilon > 0$, asymptotically almost surely. Thus the estimate
$$
\sup_{t \in \R} \left[ \frac 1n \sum_{x\in G_n}  \left| \tilde \mu_{n,x}((-\infty,t]) - \tilde \sigma_0((-\infty,t])  \right| \right]  \leq \frac{4(1+2\epsilon)}{\kappa}  \frac {\ln (d_n-1)}{\ln( n)} 
$$
holds, for all $\epsilon > 0$, asymptotically almost surely. Here we used the fact that $\tilde \gamma_{d_n}$ tends to $2/\pi$ as $d_n$ tends to infinity.
It remains to note that 
\begin{equation}
\label{scestimate}
\sup_{t \in \R} \left[ \left| \tilde \sigma_0 \lk (-\infty,t]\rk  - \sigma_\s \lk (-\infty,t]  \rk \right| \right] \leq (1+\epsilon)  \frac 2{\pi d_n}
\end{equation}
for $d_n$ large enough. Thus applying the triangle inequality yields the first claim. By \eqref{growmeasures} the second claim follows from the first.
\end{proof}


\section{Estimates for the Green function and delocalization of eigenvectors}
\label{sec:eigenfunctions}

Here we  apply the results of the previous sections to compare Green functions on the finite graph $G_n$ with the respective Green functions on the infinite tree. Then we deduce delocalization of eigenvectors for the adjacency matrix and for random Schr\"odinger operators. In particular,  we show that eigenvectors of the operator $H_n(V)$ with eigenvalues in the regime of absolutely continuous spectrum of the infinite-volume operator $H_{\mathcal T_d}(V)$ are not uniformly localized as $n$ tends to infinity.  

\subsection{Convergence of the Green function}
\label{ssec:green}

Let us first introduce some notation. 
For $z  \in \C_+$ and $x \in G_n$ we consider diagonal elements of the Green function, the Stieltjes transform of the local spectral measures:
\begin{alignat}{2}
\nonumber
\Gamma_n(x,z) &=  ( \delta_x, ( A_n-z )^{-1} \delta_x )_{\ell^2(G_n)} &=& \int_\R (\lambda-z)^{-1} \mu_{n,x}(d\lambda)  \, ,\\
\nonumber
\tilde \Gamma_n(x,z) &= ( \delta_x, ( \tilde A_n-z )^{-1} \delta_x )_{\ell^2(G_n)}  &=&  \int_\R (\lambda-z)^{-1} \tilde\mu_{n,x}(d\lambda) \, , \\
\label{greenrep}
\Gamma_n(x,z;V) &= ( \delta_x, ( H_n(V)-z )^{-1} \delta_x )_{\ell^2(G_n)}  \ &=& \int_\R (\lambda-z)^{-1} \mu_{n,x}(d\lambda;V) \, .
\end{alignat}

In the same way, we  define the corresponding Green function on the infinite tree $\mathcal T_d$. One can either use resolvent expansions and the geometric structure of the tree or the explicit representations of the measures $\sigma_0$ and $\sigma_\s$ (see \eqref{kesten} and \eqref{sc}) to derive that, for any $\hat x \in \mathcal T_d$,
\begin{equation}
\label{treegreenrep}
\Gamma_{\mathcal T_d}(\hat x,z) = \Gamma_{\mathcal T_d}(z) = \int_\R (\lambda-z)^{-1} \sigma_0(d\lambda) = \frac{-z(d-2) - d \sqrt{z^2-4(d-1)}}{2 (z^2-d^2)}
\end{equation}
and $\lim_{n \to \infty} \tilde \Gamma_{\mathcal T_{d_n}}(\hat x,z) = \Gamma_\s(z)$
with
$$
\Gamma_\s(z) = \int_\R (\lambda- z)^{-1} \sigma_\s(d\lambda) = - \frac 12 \lk z - \sqrt{z^2-4} \rk \, .
$$
Here we specify the root of a complex number as the one with positive imaginary part. As in Theorem \ref{thm:esd}, for given $x \in G_n$, we choose $\hat x \in \mathcal T_d$ from the preimage of $x$ under the universal cover.

\begin{corollary}
\label{cor:green}
Let $(\varepsilon_n)_{n \in \N}$ be a sequence of positive numbers such that $\varepsilon_n = o(1)$ as $n \to \infty$. Then the following estimates hold asymptotically almost surely.

\begin{enumerate}

\item
Assume that $(z_n)_{n \in \N}$ is a sequence of complex numbers with $\Im z_n \geq C (\varepsilon_n \log_{d-1}(n) )^{-1}$ for a constant $C > 16 \pi \gamma_d \sqrt{d-1}$.  Then we have
$$
\frac 1n \sum_{x \in G_n} \left| \Im \Gamma_n(x,z_n) - \Im \Gamma_{\mathcal T_d} (z_n) \right| \leq \varepsilon_n \, .
$$

\item Let $d_n \to \infty$ as $n \to \infty$  with $d_n \leq (n/ \ln(n))^{1/3}$ and assume that $(z_n)_{n \in \N}$ is a sequence of complex numbers with  $\Im z_n \geq C \varepsilon_n^{-1} \lk \log_{d_n-1}^{-1}(n) +  d_n^{-1} \rk$
for a constant $C > 16$.  Then we have
$$
\frac 1n \sum_{x \in G_n} | \Im \tilde \Gamma_n(x,z_n) - \Im \Gamma_\s(z_n) | \leq \varepsilon_n \, .
$$

\item Let the random potential $V$ satisfy the conditions of Theorem \ref{thm:esd} and assume 
 that $(z_n)_{n \in \N}$ is a sequence of complex numbers with $\Im z_n \geq C (\varepsilon_n \log_{d-1}(n) )^{-1}$ for a constant $C > 8\pi \| \rho \|_\infty \lk 2 \sqrt{d-1} + \rho_0 \rk$.  Then we have
$$
\frac 1n \sum_{x \in G_n} \E \left|\Im \Gamma_n(x,z_n;V) -  \Im \Gamma_{\mathcal T_d} (\hat x,z_n;V) \right| \leq \varepsilon_n \, .
$$
\end{enumerate}
\end{corollary}

\begin{proof}
Let us show how the first statement can be deduced from Theorem \ref{thm:adj}. The other two statements follow in the same way from Theorem \ref{thm:grow} and Theorem \ref{thm:esd} respectively.  

For $z \in \C_+$ write $z = E+i\eta$, with $E \in \R$ and $\eta > 0$, and 
\begin{equation}
\label{eq:gfunc}
g_{E,\eta}(\lambda) = \Im (\lambda-z)^{-1} = \frac \eta{(E-\lambda)^2 + \eta^2} \, .
\end{equation}
We note that
\begin{equation}
\label{poissonest}
\int_\R \left| \frac{\partial g_{E,\eta}}{\partial \lambda}(\lambda)  \right| d\lambda = \frac 2 \eta \, .
\end{equation}
By \eqref{greenrep} we have
$$
\Im \Gamma_n(x,z_n) = \int_\R g_{E_n,\eta_n}(\lambda) \mu_{n,x}(d\lambda) =  - \int_\R \mu_{n,x}\lk (-\infty,\lambda] \rk dg_{E_n,\eta_n}(\lambda)
$$
and by \eqref{treegreenrep}
$$
\Im \Gamma_{\mathcal T_d}(z_n) = - \int_\R \sigma_0\lk (-\infty,\lambda] \rk dg_{E_n, \eta_n}(\lambda) \, .
$$
Combining these identities with Theorem \ref{thm:adj} and with \eqref{poissonest} yields the claim.
\end{proof}

\begin{rem}
The third statement of the corollary gives a partial answer to a question raised in \cite{Fro11}, whether 
$$
\E \left| \Gamma_n (0, E + i \eta_n;V) - \Gamma_{\mathcal T_d} (0, E + i \eta_n;V) \right|  \to 0
$$
for a sequence of positive numbers $\eta_n$ that is of order $o(1)$ as $n \to \infty$. At least for the imaginary part of the Green function, the third statement of Corollary \ref{cor:green} can be interpreted in this way, if the vertex $0$ is chosen at random with uniform probability from $G_n$.
\end{rem}


\subsection{Delocalization of eigenvectors}

In this subsection we analyze how the behavior of eigenvectors of the finite-volume operator $H_n(V)$ is related to spectral properties of the infinite-volume operator $H_{\mathcal T_d}(V)$.
In the regime of absolutely continuous spectrum of the infinite-volume operator the corresponding generalized eigenfunctions are delocalized: They are not square-summable and in particular not localized to a bounded set (see for example \cite{FigNeb91} for the behavior of eigenfunctions of the adjacency matrix on an infinite regular tree and \cite{Kle98,AizWar13} for random Schr\"odinger operators).

In finite volume the spectrum is of course always pure point and eigenvectors are square-summable. In the following we show that one can nevertheless find remnants of delocalization for eigenvectors of the finite volume operator. 

Let us first consider the rescaled adjacency matrix $\tilde A_n$ on a random regular graph with degree $d_n$ tending to infinity as $n \to \infty$. In this case the spectral measure in the limit of infinite volume $n \to \infty$ is given by the semicircle measure $\sigma_\s$ defined in \eqref{sc}.  This measure is purely absolutely continuous with bounded density. Thus eigenvectors of $\tilde A_n$ are expected to be delocalized for large $n$. Proposition \ref{pro:growdeloc} shows that this is justified. The result is similar to \cite[Theorem 3]{DumPal12} and the remark after Theorem \ref{thm:grow} applies. We include the statement and its proof because it serves as an illustration for the more involved result about eigenvectors of random Schr\"odinger operators that is given in Theorem \ref{thm:deloc} below.

The proof relies on the fact that the Stieltjes transform of an absolutely continuous measure has a uniformly bounded imaginary part. In this way absolute continuous spectrum is related to boundedness of the imaginary part of the Green function:  Since the semicircle distribution  is absolutely continuous with bounded density we find, for all $E \in \R$ and $\eta > 0$,
\begin{equation}
\label{scupper}
\Im  \Gamma_\s(E + i\eta) = \int_\R g_{E,\eta}(\lambda) \sigma_{\s}(d\lambda) \leq \frac 2\pi \int_\R g_{E,\eta}(\lambda)  d\lambda = 2\, .
\end{equation}
This uniform estimate is a crucial ingredient in the proof of the following proposition.

\begin{proposition}
\label{pro:growdeloc}
Let $d_n$ satisfy the conditions of Theorem \ref{thm:grow} and let $\Lambda_n$ be a deterministic subset from $n$ vertices with $|\Lambda_n| \leq \ln(n)$. 

Let $G_n \in \mathcal G_{n,d_n}$ be a random regular graph of degree $d_n$. Then for any $\ell^2(G_n)$-normalized eigenvector $\phi$ of the operator $\tilde A_n$ the estimate   
$$
\sum_{x \in \Lambda_n} |\phi(x)|^2 \leq C  \lk  \frac{\ln(d_n-1)}{\ln(n)} + \frac 1{d_n} \rk |\Lambda_n|
$$
holds asymptotically almost surely with a uniform constant $C>0$.
\end{proposition}

The same estimate holds asymptotically almost surely if one first chooses $G_n \in \mathcal G_{n,d}$ at random and then the subset $\Lambda_n \subset G_n$ at random with uniform probability.

\begin{proof}[Proof of Proposition \ref{pro:growdeloc}]
Let $(\varphi_j)_{j=1}^n$ and $(\lambda_j)_{j=1}^n$ denote the eigenvectors and eigenvalues of the operator $\tilde A_n$. Then for any $m \in \{ 1, \dots, n\}$ and $\eta > 0$ we estimate
\begin{equation}
\label{growdeloclower}
|\varphi_m(x)|^2 \leq \eta\sum_{j = 1}^n \frac{\eta}{(\lambda_m-\lambda_j)^2 + \eta^2} | \varphi_j(x)|^2 = \eta  \Im \tilde \Gamma_n(x,\lambda_m + i\eta) \, ,
\end{equation}
where in the last step we used \eqref{greenrep}.
To derive an upper bound on $ \Im \tilde  \Gamma_n(x,\lambda_m+i\eta)$ we compare  with the Stieltjes transform of the semicircle distribution:
\begin{align*}
\left| \Im \tilde \Gamma_n(x,\lambda_m + i \eta) - \Im \Gamma_\s (\lambda_m + i \eta) \right| 
& = \left| \int_\R g_{\lambda_m,\eta} (\lambda) d \tilde \mu_{n,x}(\lambda) - \int_\R g_{\lambda_m,\eta} (\lambda) d\sigma_\s(\lambda) \right| \\
& = \left| \int_\R \lk \tilde \mu_{n,x}((-\infty,\lambda])  -  \sigma_\s((-\infty,\lambda])  \rk dg_{\lambda_m,\eta}(\lambda)  \right| \\
& \leq \int_\R \left| \tilde \mu_{n,x}((-\infty,\lambda])  -  \sigma_\s((-\infty,\lambda]) \right|  \left| g'_{\lambda_m,\eta}(\lambda) \right|  d\lambda \, .
\end{align*}
Inserting estimates \eqref{basicgrowest}, \eqref{scestimate}, and \eqref{poissonest} yields
\begin{align*}
\left| \Im \tilde \Gamma_n(x,\lambda_m + i \eta) - \Im \Gamma_\s(\lambda_m + i \eta) \right| &\leq C \lk \frac 1{R(x)^*} + \frac 1{d_n} \rk  \int_\R  \left| g'_{\lambda_m,\eta}(\lambda) \right| d\lambda \\
&\leq \frac{2C}{\eta} \lk \frac 1{R(x)^*} + \frac 1{d_n} \rk  
\end{align*}
for all $x \in G_n$ and  $\eta > 0$ with a constant $C > 0$ independent of $x$ and $n$.
We combine this bound with \eqref{scupper} and obtain
\begin{equation}
\label{eq:imgammaest}
\Im \tilde \Gamma_n(x,\lambda_m + i \eta) \leq  2 +  \frac{2C}{\eta} \lk \frac 1{R(x)^*} + \frac 1{d_n} \rk   \, .
\end{equation}

Recall the definition of $F_n(k) \subset G_n$ from Lemma \ref{lem:goodset} and the fact that $R(x) \geq k$ for $x \in F_n(k)$. Let us assume that $\Lambda_n \subset F_n(k)$. Under this assumption we combine \eqref{eq:imgammaest} with \eqref{growdeloclower} and take the limit $\eta \da 0$.  This yields
\begin{equation}
\label{growdelocfinal}
\sum_{x \in \Lambda_n}  |\varphi_m(x)|^2 \leq 2C \lk \frac 1{k-1} + \frac 1{d_n} \rk  |\Lambda_n| 
\end{equation}
for any eigenfunction $\varphi_m$, $m \in \{ 1,\dots,n \}$, on the event $\{ \Lambda_n \subset F_n(k) \}$. 

Now we estimate the probability of this event with the help of Lemma \ref{lem:goodset}. We remark that, for $|\Lambda_n| < |F_n(k)|$,
$$
\mathcal P_{n,d} \left[ \Lambda_n \subset F_n(k) \right] =  {{n-|\Lambda_n|} \choose {|F_n(k)| - |\Lambda_n|}}  {n \choose {|F_n(k)|}} ^{-1} = \frac{(n-|\Lambda_n|)! \, |F_n(k)|!}{(|F_n(k)|- |\Lambda_n|)! \, n!} \, .
$$
For a parameter $0 < \tau_n < 1- |\Lambda_n|/n$ we introduce the event $\Omega(\tau_n,k)  = \{|F_n(k)| > n (1-\tau_n) \}$ and estimate 
$$
\mathcal P_{n,d} \left[ \Lambda_n \subset F_n(k) \right] \geq \mathcal P_{n,d} \left[ \Lambda_n \subset F_n(k) \, | \, \Omega(\tau_n,k) \right] \mathcal P_{n,d} \left[ \Omega(\tau_n,k) \right] \, .
$$ 
The first factor is bounded below by
$$
\frac{(n-|\Lambda_n|)! \, (n(1-\tau_n))!}{ (n(1-\tau_n)- |\Lambda_n|)! \, n!} \geq \lk \frac{n(1-\tau_n) - |\Lambda_n|}{n} \rk^{|\Lambda_n|} = \lk 1 - \tau_n - \frac{|\Lambda_n|}n \rk^{|\Lambda_n|}
$$
and from Lemma \ref{lem:goodset} we obtain that the second factor is bounded below by
$$
 1 - \frac 1{2n \tau_n} (d-1)^{2k} \frac{\sqrt{d-1}}{\sqrt{d-1}-1} \lk 1 + \frac 8n \lk d + \frac{k}{d} \rk \rk^{2k}  \, .
$$

Now  we choose $\tau_n$ comparable to $1/\sqrt n$ and $k = \kappa \log_{d_n-1}(n) + 1$ with $\kappa < 1/4$. Then both lower bounds tend to $1$ as $n \to \infty$ and we get
$$
\mathcal P_{n,d} \left[ \Lambda_n \subset F_n \lk \kappa \log_{d_n-1}(n) + 1 \rk \right] = 1 - o(1)
$$
as $n \to \infty$. Inserting this choice of $k$ in the bound \eqref{growdelocfinal} proves the claim.
\end{proof}

\begin{rem}
The same methods yield a similar statement for the adjacency matrix $A_n$ on a random regular graph with fixed degree. However, for this operator better results have recently been derived in  \cite{BroLin13}: Brooks and Lindenstrauss also investigate delocalization on regular graphs by comparing with the regular tree. They use the explicit representation of spherical generalized eigenfunctions on the tree and estimate the norms of certain propagation operators (also build from orthogonal polynomials). From these estimates they deduce information about eigenvectors on regular graphs. This direct comparison of eigenvectors leads to delocalization bounds that decay not   logarithmically but with a small power of $n$.

Results about the behavior of eigenvectors on regular graphs related to quantum ergodicity have also been obtained in \cite{AnaLem13}, where delocalization is tested by averaging an observable. Both results rely on explicit formulas for eigenfunctions on trees that are not available for random Schr\"odinger operators. 
\end{rem}

Let us now study eigenvectors of random Schr\"odinger operators on random regular graphs with fixed degree $d$. In this case the analysis is more complicated since the spectral measure of the corresponding infinite volume operator is not purely absolutely continuous: The spectrum of a random Schr\"odinger operator on an infinite tree can consist of different components, including absolutely continuous spectrum but also pure-point spectrum. We refer to \cite{War12} for an overview of spectral properties of the operator $H_{\mathcal T_d}(V)$,  see also Appendix~\ref{ap:wegner}, where we state selected results.

Existence of pure-point spectrum and exponential localization of the corresponding eigenfunctions of $H_{\mathcal T_d}(V)$ was proved in \cite{Aiz94}. Therefore one cannot expect that all eigenvectors of the finite volume operator $H_n(V)$ on a random regular graph are delocalized.
The existence of absolutely continuous spectrum of $H_{\mathcal T_d}(V)$ was also established, first in \cite{Kle98} and later in \cite{AizSimWar06b, FroHasSpi07} and the regime where absolutely continuous spectrum can be found was recently extended in \cite{AizWar13}. 
A relevant criterion for absolutely continuous spectrum is positivity of the imaginary part of the Green function. Thus one defines
$$
\sigma_\ac(H_{\mathcal T_d}) = \left\{ \lambda \in \R \, : \, \p \left[ \lim_{\eta \da 0} \Im \Gamma_{\mathcal T_d}(x,\lambda + i \eta;V) > 0 \right] > 0 \right\} \, .
$$
This set is independent of $x \in \mathcal T_d$, deterministic and forms the support of the absolutely continuous component of the spectrum (almost surely with respect to the random potential), see \cite{AizWar13, AizWar12}. In Theorem \ref{thm:deloc} we show that eigenvectors of the finite volume operator $H_n(V)$ with   eigenvalues  in $\sigma_\ac(H_{\mathcal T_d})$ are typically delocalized for large $n$.

An important ingredient in the proof of delocalization for the adjacency matrix alone is the fact that the limiting spectral measure is absolutely continuous with uniformly boundedness of the density. This allows for  estimate \eqref{scupper}. For random Schr\"odinger operators  the density of the limiting spectral measure $\mu_{\mathcal T_d,x}$ is given by
$\lim_{\eta \da 0} \Im \Gamma_{\mathcal T_d}(x,\lambda+i\eta;V)$. This limit exists almost everywhere and it is finite if $\lambda$ lies within the absolutely continuous spectrum \cite{Tes09,AizWar13}. However, even for compact intervals $I \subset \sigma_\ac(H_{\mathcal T_d})$ it is not clear whether $\sup_{\lambda \in I, \eta > 0}  \Im \Gamma_{\mathcal T_d}(x,\lambda+i\eta;V)$ has finite expectation. Therefore we can not treat single eigenvectors but we have to select a suitable combination as follows.

For a given realization of the random potential $V$ on the tree $\mathcal T_d$, a graph $G_n \in \mathcal G_{n,d}$, and a vertex $x_0 \in G_n$ we identify the potential on the graph with the potential on the tree as described in Section \ref{ssec:rs}. Let $(\lambda_j)_{j = 1}^n$ and $(\varphi_j)_{j=1}^n$ denote the eigenvalues and corresponding $\ell^2(G_n)$-normalized eigenvectors of the operator $H_n(V)$ and let $I \subset \R$ be bounded and measurable.
For $j = 1, 2, \dots, n$ we define non-negative coefficients
\begin{equation}
\label{eq:cdef}
c_j(x_0,I) = |\varphi_j(x_0)|^2 \ \textnormal{if} \ \lambda_j \in I \quad \textnormal{and} \quad c_j(x_0,I) = 0  \ \textnormal{if} \ \lambda_j \notin I 
\end{equation}
and we note that $\sum_{j=1}^n c_j(x_0,I) \leq 1$.

In Lemma \ref{lem:treegreen} below we prove the following estimate under the assumption $I \subset \sigma_\ac(H_{\mathcal T_d})$:  For all $\hat x \in \mathcal T_d$ and $\eta > 0$ we have
\begin{equation}
\label{eq:greenest}
\E \left[  \sum_{j=1}^n c_j(x_0,I)    \Im \Gamma_{\mathcal T_d}(\hat x,\lambda_j + i \eta;V) \right] \leq C  |I|
\end{equation}
with a constant $C > 0$ depending only on the degree $d$ and on the density of the random potential. In the proof of Theorem~\ref{thm:deloc} we use this estimate  in the same way as we used \eqref{scupper} in the proof of Proposition \ref{pro:growdeloc}.

We remark  that  an eigenvector  $\varphi_j$ of $H_n(V)$ that is localized close to the vertex $x_0$ leads to a coefficient $c_j$ of order $1$. The following result shows that this can not happen for large $n$ for eigenvectors with eigenvalues within the absolutely continuous spectrum of $H_{\mathcal T_d}$. 
Recall that we write $B_r(x_0) = \{ y \in G_n \, : \, d(x_0,y) \leq r \}$ and that $|B_r(x_0)|$ denotes the number of vertices in this neighborhood.

\begin{theorem}
\label{thm:deloc}
Assume that the density of the random potential is bounded and has bounded support. Let $I \subset \sigma_\ac(H_{\mathcal T_d})$ be bounded and measurable.
For a random regular graph $G_n \in \mathcal G_{n,d}$ choose a vertex $x_0 \in G_n$ at random with uniform probability. Let $(\varphi_j)_{j=1}^n$ denote the $\ell^2(G_n)$-normalized eigenvectors of $H_n(V)$ and let $c_j(x_0,I)$ be as defined in \eqref{eq:cdef}.  

Then for any sequence $(\varepsilon_n)_{n \in \N}$ of positive numbers with  $\varepsilon_n = o(1)$ as $n \to \infty$ the estimate
$$
\sum_{x \in B_r(x_0)} \sum_{j=1}^n c_j(x_0,I) |\varphi_j(x)|^2 \leq \frac{|B_r(x_0)|}{\varepsilon_n \sqrt{\log_{d-1}(n)}}
$$
holds asymptotically almost surely  (with respect to the random potential $V$ and the choice of graph $G_n \in \mathcal G_{n,d}$ and vertex $x_0 \in G_n$) for all $r \leq \ln \lk  \ln (n) \rk$.
\end{theorem}

\begin{proof}
First we fix a graph $G_n$ and a vertex $x_0 \in G_n$ and we choose $x \in B_r(x_0)$. 
For a given realization of the potential $V$ on the tree $\mathcal T_d$ we consider the corresponding potential on the graph $G_n$ as explained in Section~\ref{ssec:rs}.

As in \eqref{growdeloclower}, we have for all $j \in \{1, \dots, n\}$ and all $\eta > 0$
$$
|\varphi_j(x)|^2 \leq \eta \, \Im  \Gamma(x,\lambda_j + i\eta;V) 
$$
and therefore
\begin{equation}
\label{deloclower}
\E  \left[ \sum_{j=1}^n c_j |\varphi_j(x)|^2 \right] \leq \eta \,  \E \left[ \sum_{j=1}^n c_j   \Im \Gamma_n(x,\lambda_j+i\eta;V)  \right]
\end{equation}
for all $\eta > 0$. Here and in the remainder of the proof we write $c_j = c_j(x_0,I)$ for short.

To derive an upper bound on $\Im \Gamma_n(x,\lambda_j+i\eta;V) $ we compare with the Green function on the tree $\mathcal T_d$. Consider the map $\iota_{\hat x_0}$ given in \eqref{eq:iotamap} and let $\hat x \in \mathcal T_d$ be the preimage of $x \in G_n$ under this map. Recall the definition of $g_{E,\eta}$ from \eqref{eq:gfunc}. For $\eta > 0$ and any eigenvalue $\lambda_j \in I$ we have
\begin{align*}
&\left| \Im \Gamma_n(x,\lambda_j + i \eta;V) - \Im \Gamma_{\mathcal T_d}(\hat x,\lambda_j + i \eta;V) \right| \\
& = \left| \int_\R g_{\lambda_j,\eta} (\lambda) \mu_{n,x}(d\lambda;V) - \int_\R g_{\lambda_j,\eta} (\lambda) \mu_{\mathcal T_d,\hat x}(d\lambda;V) \right| \\
& \leq \int_\R \left| \mu_{n,x}((-\infty,\lambda];V)  -  \mu_{\mathcal T_d,\hat x}((-\infty,\lambda];V) \right| \, \sup_{\xi \in I} \left| g'_{\xi,\eta}(\lambda) \right|  d\lambda \, .
\end{align*}
Note that a coefficient $c_j$ is non-zero only if the corresponding eigenvalue $\lambda_j$ lies in $I$ and that the sum of the coefficients $c_j$ is bounded by one. Hence, we find 
\begin{align*}
&\E \left[ \sum_{j=1}^n c_j \left| \Im \Gamma_n(x,\lambda_j + i \eta;V) - \Im \Gamma_{\mathcal T_d}(\hat x,\lambda_j + i \eta;V) \right| \right] \\
& \leq \int_\R \E \left[ \left| \mu_{n,x}((-\infty,\lambda];V)  -  \mu_{\mathcal T_d,\hat x}((-\infty,\lambda];V) \right| \right] \, \sup_{\xi \in I} \left| g'_{\xi,\eta}(\lambda) \right|  d\lambda
\end{align*}
and Theorem \ref{thm:detest} yields the upper bound
\begin{align}
\nonumber
&\E \left[ \sum_{j=1}^n c_j \left| \Im \Gamma_n(x,\lambda_j + i \eta;V) - \Im \Gamma_{\mathcal T_d}(\hat x,\lambda_j + i \eta;V) \right| \right] \\
\label{eq:greencomp}
& \leq 
2\pi \|\rho\|_\infty \lk 2 \sqrt{d-1} + \rho_0 \rk \frac 1{R(x)^*} \int_\R \lk \sup_{\xi \in I} \left| g'_{\xi,\eta}(\lambda) \right| \rk d\lambda   \leq C  \lk 1+ \frac {|I|} \eta \rk \frac 1{\eta R^*(x)} \, ,
\end{align}
where we used the fact that the integral of the supremum is bounded by a  constant times $\eta^{-1}(1+|I|\eta^{-1})$. Here $\rho$ denotes the density of the random potential and we use that $\| \rho \|_\infty < \infty$ and that $\supp \rho = [-\rho_0, \rho_0]$ with $\rho_0 < \infty$. To shorten notation $C$ denotes various positive constants that may depend only on $d$ and $\rho$.

Combining \eqref{deloclower} and \eqref{eq:greencomp} with \eqref{eq:greenest} yields, for $\eta > 0$ small enough,
$$
\E \left[ \sum_{j=1}^n c_j |\varphi_j(x)|^2 \right] \leq C  \lk |I| \eta + \frac 1{R(x)^*} + \frac{|I|}{\eta R(x)^*} \rk \leq C |I| \lk \eta + \frac{1}{\eta R(x)^*} \rk \, .
$$
Applying the Markov inequality we deduce that for any $0 < \delta_n < 1$
\begin{equation}
\label{initialprob}
\p \left[ \sum_{x \in B_r(x_0)}   \sum_{j=1}^n c_j |\varphi_j(x)|^2 \leq \frac{C |I|}{\delta_n} \sum_{x \in B_r(x_0)} \lk \eta + \frac{1}{\eta R(x)^*} \rk  \right] \geq 1- \delta_n \, .
\end{equation}

It remains to estimate $R(x)^*$. 
Recall the definition of  $F_n(k) \subset G_n$ from Lemma~\ref{lem:goodset} and assume that $x_0 \in F_n(k)$.  Since $x_0$ is chosen from $G_n$ at random with uniform probability we can estimate the probability of this event with the help of Lemma~\ref{lem:goodset}. Indeed, we have, for any $0 < \tau_n < 1$,
\begin{align*}
\mathcal P_{n,d} \left[ x_0 \in F_n(k) \right] &\geq \mathcal P_{n,d} \left[ x_0 \in F_n(k) \, | \, |F_n(k)| \geq n(1-\tau_n) \right] \ \mathcal P_{n,d} \left[ |F_n(k)| \geq n(1-\tau_n) \right] \\
&\geq (1-\tau_n)  \mathcal P_{n,d} \left[ |F_n(k)| \geq n(1-\tau_n) \right] \, .
\end{align*}
By Lemma \ref{lem:goodset} the latter probability is bounded below by
$$
1 - \frac 1{2n \tau_n} (d-1)^{2k} \frac{\sqrt{d-1}}{\sqrt{d-1}-1} \lk 1 + \frac 8n \lk d + \frac{k}{d} \rk \rk^{2k} \, .
$$
Let us now choose $k = \kappa \log_{d-1}(n) + 1$ with $\kappa < 1/4$ and $\tau_n$ comparable to $1/\sqrt n$. Then the lower bound tends to $1$ as $n \to \infty$ and we get
\begin{equation}
\label{graphproblower}
\mathcal P_{n,d} \left[ x_0 \in F_n \lk \kappa  \log_{d-1}(n) +1 \rk \right] = 1 - o(1)
\end{equation}
as $n \to \infty$. 

By assumption, $r \leq \ln \lk \ln (n) \rk < \kappa \log_{d-1}(n)+1$ for $n$ large enough. Hence, for $x \in B_r(x_0)$ we have
\[
B_{\kappa \log_{d-1}(n)+1-r}(x) \subset B_{\kappa \log_{d-1}(n)+1}(x_0) \, . 
\]
Thus $x_0 \in F_n(\kappa \log_{d-1}(n)+1)$ implies that these neighborhoods are acyclic. In particular we find $x \in F_n(\kappa \log_{d-1}(n)-r+1)$ so that  $R(x)^* \geq \kappa \log_{d-1}(n)-r$. We choose $\eta = (\kappa \log_{d-1}(n) -r)^{-1/2}$ and arrive at
$$
\sum_{x \in B_r(x_0)} \lk \eta + \frac{1}{\eta R(x)^*} \rk \leq \frac{2 |B_r(x_0)|}{\sqrt {\kappa \log_{d-1}(n)-r}} \leq \frac{4 |B_r(x_0)|}{\sqrt {\kappa \log_{d-1}(n)}} 
$$
for $n$ large enough.
  
With this choice of parameters \eqref{initialprob} reads as
\[
\p \left[ \sum_{x \in B_r(x_0)}   \sum_{j=1}^n c_j |\varphi_j(x)|^2 \leq  \frac{C |I|}{\delta_n} \frac{|B_r(x_0)|}{\sqrt{\kappa \log_{d-1}(n)}}  \right] \geq 1- \delta_n 
\]
and this estimate is valid for $n$ large enough on the event $\{x_0 \in F_n(\kappa \log_{d-1}(n) +1) \}$.
Finally, we choose $\delta_n$ comparable to $C |I| \varepsilon_n / \sqrt \kappa$ such that 
\begin{equation}
\label{finaldelocest}
\p \left[ \sum_{x \in B_r(x_0)}   \sum_{j=1}^n c_j |\varphi_j(x)|^2 \leq  \frac{|B_r(x_0)|}{\varepsilon_n \sqrt{  \log_{d-1}(n)}}  \right] \geq 1- \frac{C|I|}{\sqrt \kappa} \, \varepsilon_n 
\end{equation}
for $n$ large enough on the event  $\{x_0 \in F_n(\kappa \log_{d-1}(n) +1) \}$.
Relations \eqref{graphproblower} and \eqref{finaldelocest} show that the claimed estimate holds asymptotically almost surely and the proof is complete.
\end{proof}

The proof of Theorem \ref{thm:deloc} relies on the following estimate for the Green function on the tree. For this estimate it is essential that $I \subset \sigma_{\ac}(H_{\mathcal T_d})$.

\begin{lemma}
\label{lem:treegreen}
Under the conditions of Theorem \ref{thm:deloc} the bound \eqref{eq:greenest} holds for all $\hat x \in \mathcal T_d$ and $\eta > 0$ with a constant $C > 0$ depending only on the degree $d$ and on the density of the random potential.
\end{lemma}

\begin{proof}
The proof  is based on estimate \eqref{invmoments}, the fact that within the absolutely continuous spectrum the imaginary part of the Green function has finite inverse moments. This was recently proved in \cite{AizWar12}. To apply this result we rely on recursion properties of the Green function on trees and on results from rank-one perturbation theory. 

We fix a graph $G_n$ and a vertex $x_0 \in G_n$. 
Recall the definition of the coefficients $c_j(x_0,I)$ from \eqref{eq:cdef}. Again we write $c_j = c_j(x_0,I)$ for short. By \eqref{eq:randmes} we have, for any $\hat x \in \mathcal T_d$, 
\begin{align*}
\sum_{j = 1}^n c_j \Im \Gamma_{\mathcal T_d}(\hat x, \lambda_j + i \eta;V) &= \sum_{  \lambda_j \in I } |\varphi_j(x_0)|^2 \Im \Gamma_{\mathcal T_d}(\hat x, \lambda_j + i \eta;V) \\
&= \int_I  \Im \Gamma_{\mathcal T_d}(\hat x, \lambda + i \eta;V) \mu_{n,x_0}(d\lambda;V) \, .
\end{align*}
We use the Stieltjes inversion formula, see for example \cite[Theorem 3.21]{Tes09}, to write the spectral measure $\mu_{n,x_0}$ in terms of the Green function and we obtain
\begin{equation}
\label{eq:stieltjes}
\sum_{j = 1}^n c_j \Im \Gamma_{\mathcal T_d}(\hat x, \lambda_j + i \eta;V) = \lim_{\epsilon \da 0} \frac 1 \pi \int_I  \Im \Gamma_{\mathcal T_d}(\hat x, \lambda + i \eta;V) \, \Im \Gamma_{n}(x_0,\lambda+i\epsilon;V) \, d\lambda \, .
\end{equation}
To estimate the expectation of the right-hand side, we use the following results from rank-one perturbation theory \cite{DeyLevSou85,SimWol86,Aiz94}. This allows to analyze the dependence on the single random variable $\omega_{x_0}$.

First we rewrite $ \Im \Gamma_{n}(x_0,\lambda+i\epsilon;V)$:
For a realization of the random potential $V=(\omega_x)_{x \in G_n}$
we denote by $\hat V$ the same collection of random variables with $\omega_{x_0}$ replaced by zero. Then 
$$
H_n(V) = H_n(\hat V) + \omega_{x_0} \delta_{x_0}  \, ,
$$
where $\delta_{x_0}(x_0) = 1$  and $\delta_{x_0}(x) = 0$ for $x \neq x_0$.  The resolvent identity yields, for  $z \in \C_+$,
\[
\frac 1{H_n(V) -z} - \frac 1{H_n(\hat V)-z} = \frac 1{H_n(V) -z} \, \omega_{x_0} \delta_{x_0} \, \frac 1{H_n(\hat V)-z}
\]
and thus
\[
\Gamma_n(x_0,z;V) = \frac{\Gamma_n(x_0,z;\hat V)}{1+\omega_{x_0} \Gamma_n(x_0,z;\hat V)} = \frac 1 {\omega_{x_0} - \Xi_{x_0}(z)}
\]
with 
\[
\Xi_{x_0}(z) =  - \Gamma_n(x_0,z;\hat V)^{-1} \, .
\]
It follows that
\begin{equation}
\label{eq:rankonerep}
\Im \Gamma_n(x_0, \lambda+i\epsilon; V) = \frac{\Im \Xi_{x_0}(\lambda+i\epsilon)}{(\omega_{x_0}-\Re \Xi_{x_0}(\lambda+i\epsilon) )^2 + \Im \Xi_{x_0}(\lambda+i\epsilon)^2} \, .
\end{equation}
We emphasize that $\Xi_{x_0}$ is independent of $\omega_{x_0}$. We also note that the limit $\Xi_{x_0}(\lambda) = \lim_{\epsilon \da 0} \Xi_{x_0}(\lambda+i\epsilon)$ exists almost everywhere with respect to Lebesgue measure, see for example \cite[Theorem 3.23]{Tes09}.

Next, we estimate $ \Im \Gamma_{\mathcal T_d}(\hat x, \lambda + i \eta;V)$:
If we remove the vertex $\hat x$ from the tree $\mathcal T_d$, it is decomposed into $d$ disjoint infinite rooted trees that are rooted at the the nearest neighbors of $\hat x$.   (We remark that these trees are no longer regular, since the degree at the root equals $d-1$.) Let $N(\hat x) = \{ y \in \mathcal T_d \, : \, d(\hat x,y)=1 \}$ denote the set of nearest neighbors of $\hat x$. For $y \in N(\hat x)$, let $T_y$ denote the rooted tree with root at $y$. In the same way as on the regular tree we define the bounded self-adjoint operators $H_{T_y}(V)$ with domain $\ell^2(T_y)$. For $u \in T_y$ and $z \in \C_+$ we write 
\[
\Gamma_{T_y}(u,z;V) = \lk \delta_u, \lk H_{T_y}(V) - z \rk^{-1} \delta_u \rk_{\ell^2(T_y)}
\]
for the respective Green function. The trees $T_y$ with $y \in N(\hat x)$ are not connected to each other, hence the random variables $\Gamma_{T_y}(y,z;V)$, $y \in N(\hat x)$, are independent.

Employing the resolvent equation one can derive the recursion formula
$$
\Gamma_{\mathcal T_d}(\hat x,z;V) = \frac 1{\omega_{\hat x} - z - \sum_{y \in N(\hat x)} \Gamma_{T_y}(y,z;V) }\, .
$$
Taking the imaginary part, we see that for all $z \in \C_+$
$$
\Im \Gamma_{\mathcal T_d}(\hat x, z; V) = |\Gamma_{\mathcal T_d}(\hat x,z;V)|^2 \lk \sum_{y \in N(\hat x)}  \Im \Gamma_{T_y}(y,z;V) + \Im z \rk \, .
$$
Applying the recursion formula one also gets that 
\[
|\Gamma_{\mathcal T_d}(\hat x,z;V)|^2   \leq \frac{1}{\lk \Im \lk \omega_{\hat x} - z - \sum_{y \in N(\hat x)} \Gamma_{T_y}(y,z;V) \rk \rk^2 } = \frac{1}{\lk \sum_{y \in N(\hat x)}   \Im \Gamma_{T_y}(y,z;V) + \Im z  \rk^2}
\]
and we obtain
\begin{equation}
\label{eq:imestinv}
\Im \Gamma_{\mathcal T_d}(\hat x, z ; V)  \leq  \frac1{\sum_{y \in N(\hat x)}  \Im \Gamma_{T_y}(y,z;V) } \, .
\end{equation}

Let again $\hat x_0 \in \mathcal T_d$ denote the vertex corresponding to $x_0 \in G_n$ under the universal cover and let us for the moment assume that $\hat x \neq \hat x_0$. Then there is one unique vertex $y^* \in N(\hat x)$ such that the corresponding rooted tree $T_{y^*}$ contains $\hat x_0$.
Let $N^+(\hat x) = \{ y \in N(\hat x) \, : \, \hat x_0 \notin T_y\}$ denote the subset of all other nearest neighbors of $\hat x$.  In particular, we find that $N^+(\hat x)$ contains $d-1$ vertices. For $y \in N^+(\hat x)$ the tree $T_y$ does not contain $\hat x_0$, thus  the collection of random variables
\[
\lk \Gamma_{T_y}(y,\lambda+i\eta;V) \rk_{y \in N^+(\hat x)} 
\]
is independent of $\omega_{\hat x_0}$, the value of the random potential $V$ at $\hat x_0$.
If $\hat x = \hat x_0$, then all random variables $\Gamma_{T_y}(y,\lambda+i\eta,V)$ with $y \in N(\hat x)$ are independent of $\omega_{\hat x_0}$ and we can continue the proof in the same way with $N^+(\hat x)$ replaced by $N(\hat x)$.

We note that, for $z \in \C_+$, the imaginary part of the Green function is positive. Hence, in view of \eqref{eq:imestinv} we have
\begin{equation}
\label{eq:imestinv2}
\Im \Gamma_{\mathcal T_d}(\hat x, z; V)  \leq  \frac1{ \sum_{y \in N^+(\hat x)}  \Im \Gamma_{T_y}(y,z;V) } \, .
\end{equation}
Combining \eqref{eq:stieltjes}, \eqref{eq:rankonerep}, and \eqref{eq:imestinv2} we arrive at the bound
\begin{align*}
&\sum_{j = 1}^n c_j \Im \Gamma_{\mathcal T_d}(\hat x, \lambda_j + i \eta;V) \\
& \leq \lim_{\epsilon \da 0} \frac 1 \pi \int_I \frac {1 }{\sum_{y \in N^+(\hat x)}  \Im \Gamma_{T_y}(y,\lambda+i\eta;V)  } \frac{\Im \Xi_{x_0}(\lambda+i\epsilon)}{(\omega_{x_0}-\Re \Xi_{x_0}(\lambda+i\epsilon) )^2 + \Im \Xi_{x_0}(\lambda+i\epsilon)^2} d\lambda \, .
\end{align*}

Recall that $\omega_{\hat x_0} = \omega_{x_0}$ and that the random variables $\Xi_{x_0}(\lambda+i\epsilon)$ and $\Gamma_{T_y}(y,\lambda+i\eta;V)$ with $y \in N^+(\hat x)$ are independent of $\omega_{\hat x_0}$. Hence, we condition on the random potential $V = (\omega_y)_{y \in \mathcal T_d}$ at all other vertices and take the expectation with respect to $\omega_{\hat x_0}$ only. By dominated convergence, we find
\begin{align*}
&\E \left[ \left. \sum_{j = 1}^n c_j \Im \Gamma_{\mathcal T_d}(\hat x, \lambda_j + i \eta;V)  \right|  \lk \omega_y \rk_{y \neq \hat x_0} \right]\\
&\leq \lim_{\epsilon \da 0}  \frac 1 \pi \int_\R \int_I \frac {1}{\sum_{y \in N^+(\hat x)}  \Im \Gamma_{T_y}(y,\lambda+i\eta;V)  }  \frac{\Im \Xi_{x_0}(\lambda+i\epsilon)}{(v-\Re \Xi_{x_0}(\lambda+i\epsilon) )^2 + \Im \Xi_{x_0}(\lambda+i\epsilon)^2}  d\lambda \, \rho(v) \, dv \, .
\end{align*}
Now we estimate $\rho(v) \leq \| \rho \|_\infty$ and apply Fubini's theorem to perform the integration in $v$. Since
\[
\int_\R  \frac{\Im \Xi_{x_0}(\lambda+i\epsilon)}{(v-\Re \Xi_{x_0}(\lambda+i\epsilon) )^2 + \Im \Xi_{x_0}(\lambda+i\epsilon)^2} dv = \pi
 \] 
we obtain
$$
\E \left[ \left. \sum_{j = 1}^n c_j \Im \Gamma_{\mathcal T_d}(\hat x, \lambda_j + i \eta;V)  \right|  \lk \omega_y \rk_{y \neq \hat x_0} \right] \leq \| \rho \|_\infty
\int_I   \frac {1}{\sum_{y \in N^+(\hat x)}  \Im \Gamma_{T_y}(y,\lambda+i\eta;V)  } \, d\lambda \, . 
$$
Taking now the expectation with respect to $(\omega_y)_{y \neq \hat x_0}$ yields
$$
\E \left[ \sum_{j = 1}^n c_j \Im \Gamma_{\mathcal T_d}(\hat x, \lambda_j + i \eta,V) \right] \leq  \| \rho \|_\infty \int_I   \E \left[ \frac {1}{\sum_{y \in N^+(\hat x)}  \Im \Gamma_{T_y}(y,\lambda+i\eta,V)  } \right] d\lambda \, .
$$
Finally, we use that the set $N^+(x)$ contains $d-1$ elements. Jensen's inequality tells us that
$$
\sum_{y \in N^+(\hat x)}  \Im \Gamma_{T_y}(y,\lambda+i\eta;V) \geq (d-1) \prod_{y \in N^+(\hat x)} \Im \Gamma_{T_y}(y,\lambda+i\eta;V) ^{1/(d-1)}
$$
and using the fact that the random variables  $\Im \Gamma_{T_y}(y,\lambda+i\eta;V)$, $y \in N^+(\hat x)$, are independent we conclude 
$$
\E \left[ \sum_{j = 1}^n c_j \Im \Gamma_{\mathcal T_d}(\hat x, \lambda_j + i \eta;V) \right] \leq \frac{ \| \rho \|_\infty}{d-1} \int_I \prod_{y \in N^+(\hat x)}\E \left[ \Im \Gamma_{T_y}(y,\lambda+i\eta;V)^{-1/(d-1)} \right] d\lambda \, .
$$ 
Hence, applying \eqref{invmoments} yields the claim of the lemma.
\end{proof}


\appendix

\section{Auxiliary results}

In the appendix we collect some known results that are used in the previous sections. We indicate where proofs can be found in the literature.

\subsection{Orthonormal polynomials and Gaussian quadrature}
\label{ap:quadrature}

In Section \ref{sec:detest} we repeatedly used the following quadrature formula for polynomials. A proof and further references can be found for example in \cite[Chapter 1.4.1]{Akh65}.

\begin{lemma}
\label{lem:quadrature}
Let $\sigma$ be a measure on the real line with finite moments and let $(P_n)_{n \in \N_0}$ denote the orthonormal polynomials with respect to $\sigma$. 
For arbitrary $M \in \N$ and $s \in \R$ set $\hat P_{M+1} = P_{M+1} + s P_M$ and let $\lambda_1 < \lambda_2 < \dots < \lambda_{M+1}$ denote the zeros of $\hat P_{M+1}$.

Then the identity
$$
\int_\R R(\lambda) \sigma(d\lambda) = \sum_{k=1}^{M+1} \frac{R(\lambda_k)}{\sum_{n=0}^M P_n(\lambda_k)} 
$$
holds for any polynomial $R$ of degree less or equal than $2M$. 
\end{lemma}

\subsection{Spectrum and Green function on the infinite tree}
\label{ap:wegner}

Here we mention some results about the spectrum of the random Schr\"odinger operator $H_{\mathcal T_d}(V)$ defined in Section \ref{ssec:rs} and about the Green function $\Gamma_{\mathcal T_d}(x,z;V)$ defined in Section \ref{ssec:green}. We refer to the books \cite{Kir89,CarLac90,PasFig92,Tes09} for more information and further references.

Recall that the random potential $V$ is defined as a multiplication operator
$$
\lk V \phi \rk (x) = \omega_x \phi(x) \, , \quad \phi \in \ell^2(\mathcal T_d) \, ,\quad x \in \mathcal T_d \, ,
$$ 
where $\lk \omega_x \rk_{x \in \mathcal T_d}$ are independent identically distributed real random variables with density $\rho$. Hence one can refer to the  theory of ergodic operators to determine the spectrum of $H_{\mathcal T_d}$. In \cite{KunSou80,KirMar82} it is shown that the spectrum corresponds almost surely to the set-sum of the spectrum of the adjacency matrix and the support of $\rho$. 
On the tree the spectrum of $A_{\mathcal T_d}$ is given by $(-2\sqrt{d-1},2\sqrt{d-1} )$. So under the assumption $\textnormal{supp}(\rho) = [-\rho_s,\rho_s]$ with $\rho_s < \infty$ the spectrum of $H_{\mathcal T_d}(V)$ is almost surely given by the deterministic set $[-2 \sqrt{d-1}-\rho_s, 2\sqrt{d-1}+\rho_s]$. In particular, the spectral measure $\mu_{\mathcal T_d,x}$ satisfies, for all $x \in \mathcal T_d$,
\begin{equation}
\label{spectralsupport}
\textnormal{supp} ( \mu_{\mathcal T_d,x}) =  [ -2 \sqrt{d-1}-\rho_s, 2\sqrt{d-1}+\rho_s ] 
\end{equation}
almost surely and this implies $\textnormal{supp} ( \sigma_\rho) =  [-2 \sqrt{d-1}-\rho_s, 2\sqrt{d-1}+\rho_s]$.

It was noticed by Wegner \cite{Weg81} that regularity of the distribution of $\omega_x$ implies regularity of the density of states measure: Under the assumption $\| \rho \|_\infty < \infty$ one has
\begin{equation}
\label{wegner}
\left\| \frac{d\sigma_\rho}{d\lambda} \right\|_\infty \leq \|\rho\|_\infty \, .
\end{equation}

Along with the spectrum also the spectral components, the pure-point spectrum, the singular continuous spectrum, and the absolutely continuous spectrum form almost surely deterministic sets. It is the subject of extensive research to determine the location of these spectral components and we refer to \cite{War12} for an overview of results and further references.

One useful criterion for absolutely continuous spectrum is that the imaginary part of the green function does not vanish. So one considers the set
$$
\sigma_\ac(H_{\mathcal T_d}) = \left\{ \lambda \in \R \, : \, \p \left[ \lim_{\eta \da 0} \Im \Gamma_{\mathcal T_d}(x,\lambda + i \eta;V) > 0 \right] > 0 \right\}
$$
that also forms a deterministic set that does not depend on $x \in \mathcal T_d$. For almost every realization of the randomness $\sigma_\ac(H_{\mathcal T_d})$ is the support of the absolutely continuous component of the spectrum \cite{AizWar13,AizWar12}.  Within this set the imaginary part of the Green function has finite inverse moments. This fact was recently established in \cite[Theorem 2.4]{AizWar12}: Let $I \subset \sigma_{\ac}(H_{\mathcal T_d})$ be a bounded and measurable set. Consider an infinite rooted tree $T$ and let $0$ denote the vertex at the root. Then there is a $\delta > 0$ such that the estimate
\begin{equation}
\label{invmoments}
\esssup_{\lambda \in I, \eta > 0} \E \left[ (\Im \Gamma_{T}(0,\lambda+i \eta;V))^{-3-\delta} \right] < \infty  
\end{equation}
holds.

\subsection*{Acknowledgments} The author wants to thank Michael Aizenman, Sasha Sodin and Philippe Sosoe for helpful discussions. Financial support from DFG grant GE 2369/1-1 and NSF grant PHY-1122309 is gratefully acknowledged.




\end{document}